\documentclass[final]{article}

\usepackage{microtype}
\usepackage{amsmath}
\usepackage{graphicx}
\usepackage{xcolor}
\usepackage{soul}
\usepackage{algorithmic}
\usepackage{algorithm}
\usepackage{tabto}
\usepackage{caption}
\usepackage{subcaption}
\usepackage{amsthm}

\newtheorem{theorem}{Theorem}
\newtheorem{lemma}{Lemma}

\newtheorem{corollary}{Corollary}

\newcommand{\starinst}[5]{\mathcal{#1}_{#2}(#2,~#3,~#4,~#5)}	

\newcommand{\facilityset}{\mathcal{F}}	
	
\newcommand{\facset}{\mathcal{F}}				

\newcommand{\clientset}{\mathcal{C}}			
\newcommand{\cliset}{\mathcal{C}}				

\newcommand{\budget}{\mathcal{B}}

\newcommand{\capacity}{\textit{u}}				

\newcommand{\crich}{\mathcal{C}_D}				
\newcommand{\cpoor}{\mathcal{C}_S}				
\newcommand{\cdense}{\mathcal{C}_D}				
\newcommand{\csparse}{\mathcal{C}_S}

\newcommand{\ballofj}[1]{\textit{ball($ #1 $)}} 	

\newcommand{\C}[1]{\hat{C_{#1}}}					
\newcommand{\hatofC}[1]{\hat{C_{#1}}}				

\newcommand{\bundle}[1]{\mathcal{N}_{#1}}			
\newcommand{\neighbor}[1]{\mathcal{N}_{#1}}

\newcommand{\R}[1]{\mathcal{R}_{#1}}					
\newcommand{\T}[1]{\tau({#1})}

\newcommand{\cen}[1]{G_{#1}}
\newcommand{\sj}[1]{\mathcal{S}_{#1}}
	
\newcommand{\dist}[2]{c(#1,~#2)}

\newcommand{\facilitycost}{\textit{$ f_i $}}		

\newcommand{\A}[2]{\mathcal{A}_{\rho^*}(#1,{#2})}						

\newcommand{\bard}[1]{{d_{#1}}}			

\newcommand{\floor}[1]{\lfloor{#1}\rfloor}	

\newcommand{\sumlimits}[2]{\displaystyle\sum\limits_{#1}^{#2}}												

\newcommand{\sumap}[1]{\sum_{#1}{}} 

\newcommand{\price}[1]{b^c_{#1}}
\newcommand{\pricef}[1]{b^f_{#1}}

\newcommand{\load}[1]{l_{#1}}

\newcommand{\singleclient}{\textit{j}}
\newcommand{\singlefacility}{\textit{i}}

\newcommand{\zofi}[1]{z_#1}

\newcommand{\primezofi}[1]{z'_{#1}}

\newcommand{\demandofj}[1]{\textit{$ d_{#1} $}}

\newcommand{\obj}[1]{\mathcal{C}ostKM(#1)}

\newcommand{\cardTwo}[2]{size(#1,~#2)}

\newcommand{\etal}{\textit{et al}.}

\newcommand{\soner}[1]{s^1_{#1}}
\newcommand{\stwor}[1]{s^2_{#1}}
\newcommand{\Soner}[1]{S^1_{#1}}
\newcommand{\Stwor}[1]{S^2_{#1}}

\newcommand{\mcone}[1]{G^1_{#1}}
\newcommand{\mctwo}[1]{G^2_{#1}}

\newcommand{\dense}{j_d}

\newcommand{\sparseone}{j_{s}}

\newcommand{\resj}[1]{res({#1})}
\newcommand{\floordjbyu}[1]{\lfloor{d_{#1}/u}\rfloor}

\newcommand{\tauhat}[1]{\hat{\tau}({#1})}

\newcommand{\sigmaone}{\psi}
\newcommand{\sr}[1]{s_{#1}}
\newcommand{\Sr}[1]{S_{#1}}

  \date{}
\begin{document}

	\title{Constant factor Approximation Algorithm for Uniform Hard Capacitated Knapsack Median Problem}
	
	\maketitle
  
	\begin{center}
		\author{Sapna Grover$^1$,}
		\author{Neelima Gupta$^2$,}
		\author{Samir Khuller$^3$ and }
		\author{Aditya Pancholi$^4$}
	\end{center}
	\begin{enumerate}
		\item {Department of Computer Science, University of Delhi, India.\\
			\texttt{sgrover@cs.du.ac.in,sapna.grover5@gmail.com}}
		\item {Department of Computer Science, University of Delhi, India.\\
			\texttt{ngupta@cs.du.ac.in}}
		\item {Department of Computer Science, University of Maryland, USA. \\
		\texttt{samir@cs.umd.edu}}
		\item {Department of Computer Science, University of Delhi, India.\\
			\texttt{apancholi@cs.du.ac.in}}
	\end{enumerate}







%
	\begin{abstract}
		In this paper, we give the first constant factor approximation algorithm for capacitated knapsack median problem (CKnM) for hard uniform capacities, violating the budget by a factor of $1+\epsilon$ and capacities by a $2+\epsilon$ factor. To the best of our knowledge, no constant factor approximation is known for the problem even with capacity/budget/both violations. Even for the uncapacitated variant of the problem, the natural LP is known to have an unbounded integrality gap even after adding the covering inequalities to strengthen the LP. 
		Our techniques for CKnM provide two types of results for  the capacitated $k$-facility location problem. We present an $O(1/\epsilon^2)$ factor approximation for the problem, violating capacities by $(2+\epsilon)$. Another result is an $O(1/\epsilon)$ factor approximation, violating the capacities by a factor of at most $(1 + \epsilon)$ using at most $2k$ facilities for a fixed $\epsilon>0$.
		As a by-product, a constant factor approximation algorithm for capacitated facility location problem with uniform capacities is presented, violating the capacities 
		by ($1 + \epsilon$) factor. Though constant factor results are known for the problem without violating the capacities, the result is interesting as it is obtained by rounding the solution to the natural LP, which is known to have an unbounded integrality gap without violating the capacities. Thus, we achieve the best possible from the natural LP for the problem. The result shows that the natural LP is not too bad.								
		
		\textbf{keywords:} Capacitated Knapsack Median, Capacitated $k$ -Facility Location
		
	\end{abstract}
	
\section{Introduction}

\label{intro}

Facility location  and $k$-median problems are well studied in the literature. 
In this paper, we study some of their generalizations. In particular, we study capacitated variants of the knapsack median problem (KnM) and the $k$ facility location problem ($k$FLP).
Knapsack median problem is a generalization of the $k$-median problem, in which we are given a set $\clientset$ of clients with demands, a set $\facilityset$ of facility locations and a budget $\budget$. 
Setting up a facility at location $i$ incurs cost $f_i$ (called the {\em facility opening cost} or simply the {\em facility cost} ) and servicing a
client $j$ by a facility $i$ incurs cost $\dist{i}{j}$ (called the {\em service cost}). 
We assume that the costs are metric i.e., they satisfy the triangle inequality.
The goal is to select the locations to install facilities, so that the total cost for setting up the facilities does not exceed $\budget$ and  the cost of servicing all the clients by the opened facilities is minimized. When $f_i =1\ \forall i \in \facilityset$ and $\budget =k$, it reduces to the $k$-median problem.
In the {\em capacitated} version of the problem, we are also given a bound $\capacity_i$ on the maximum number of clients that facility $i$ can serve.  Given a set of open facilities, an assignment problem is solved to determine the best way of servicing the clients.
Thus any solution is completely determined
by the set of open facilities. In this paper, we address the capacitated knapsack median (CKnM) problem with uniform capacities i.e., $u_i = u~\forall i \in \facilityset$ and clients with unit demands.
In particular, we present the following result:

\begin{theorem}
	\label{theorem1}
	There is a polynomial time algorithm that approximates hard uniform capacitated knapsack median problem within a constant factor violating the capacity by a factor of at most $(2 + \epsilon)$ and budget by a factor of at most $(1 + \epsilon)$, for every fixed $\epsilon > 0$.
\end{theorem}

Our result is nearly the best achievable from rounding the natural LP: we cannot expect to get rid of the violation in the budget as it would imply a constant factor integrality gap for the uncapacitated case which is known to have an unbounded integrality gap. Even with budget violation, capacity violation cannot be reduced to below $2$ as it would imply less than $2$ factor capacity violation for $k$-median problem with $k + 1$ facilities. The natural LP has an unbounded integrality gap for this scenario as well\footnote{Let $M$ be a large integer, $u_i=M$ and $k = 2 M - 2$. There are  $M$ groups of locations; distance between locations within a group is $0$ and distance between locations in two different groups is $1$. Each group has $2 M - 2$ facilities and $2 M - 2$ clients, all co-located. In an optimal LP solution each facility is opened to an extent of $1/M$ thereby creating a capacity of $2 M - 2$ within each group. In an integer solution, if at most $k + 1 = 2 M - 1$ facilities are allowed to be opened then there is at least one group with only one  facility opened in it. Thus capacity in the group is $M$ whereas the demand is $2 M - 2$. Thus the blowup in capacity is $(2M - 2)/M$.} \footnote{We thank Moses Charikar for providing the above example where violation in one of the parameters is less than $2$ factor and no violation in the other. The example was subsequently modified by us to allow $k + 1$ facilities.}.

The $k$-facility location problem ($k$FLP) is a common generalization of the facility location problem and the $k$-median problem. In $k$FLP,  we are given a bound $k$ on the maximum number of facilities that can be opened (instead of a budget on the total facility opening cost) and the objective is to minimize the total of facility opening cost and the cost of servicing the clients by the opened facilities.
In particular we present the following two results:

\begin{theorem}
	\label{theorem2}
	There is a polynomial time algorithm that approximates hard uniform capacitated $k$-facility location problem within a constant factor $(O(1/\epsilon^2))$ violating the capacities by a factor of at most $(2 + \epsilon)$ for every fixed $\epsilon > 0$.
\end{theorem}

\begin{theorem}
	~\label{C$k$FLP-card}
	\label{theorem3}
	There is a polynomial time algorithm that approximates hard uniform capacitated $k$-facility location problem within a constant factor $(O(1/\epsilon))$ violating the capacity by a factor of at most $(1 + \epsilon)$ using at most $2k$ facilities for every fixed $\epsilon> 0$.
\end{theorem}

As a particular case of C$k$FLP, we obtain the following interesting result for the capacitated facility location problem (CFLP):

\begin{corollary}
	There is a polynomial time algorithm that approximates hard uniform capacitated facility location problem within a constant factor $(O(1/\epsilon))$ violating the capacity by a factor of at most $(1 + \epsilon)$ for every fixed $\epsilon> 0$.
\end{corollary}

The standard LP is known to have an unbounded integrality gap for CFLP even with uniform capacities. Though constant factor results are known for the problem without violating the capacities~\cite{mathp,Anfocs2014}, our result is interesting as it is obtained by rounding the solution to the natural LP.
Our result shows that the natural LP is not too bad. 

\subsection{Motivation and Challenges}

The natural LP for KnM is known to have an unbounded integrality gap~\cite{charikar2005improved} even for the uncapacitated case.
Obtaining a constant factor approximation for the (capacitated) $k$-median (CkM) problem 
is still open, let alone the CKnM problem. Existing solutions giving constant-factor approximation for CkM violate at least one of the two ($cardinality$ and $capacity$) constraints. Natural LP is known to have an unbounded integrality gap when any one of the two constraints is allowed to be violated by a factor of less than $2$ without violating the other.

Several results~\cite{charikar,Charikar:1999,capkmByrkaFRS2013,capkmshanfeili2014,KPR,capkmGijswijtL2013} have been obtained for CkM that violate either the capacities or the cardinality by a factor of $2$ or more. 
The techniques used for CkM cannot be used for CKnM as they work by transferring the opening from one facility to another (ensuring bounded service cost) facility thereby maintaining the cardinality within claimed bounds. This works well when there are no facility opening costs or the (facility opening) costs are uniform. For the general opening costs, this is a challenge as a facility, good for bounded service cost, may lead to budget violation. To the best of our knowledge, capacitated knapsack median  problem has not been addressed earlier.

C$k$FLP is NP-hard even when there is only one client and there are no facility costs ~\cite{capkmGijswijtL2013}. The hardness results for CkM hold for C$k$FLP as well. On the other hand, standard LP for capacitated facility location problem (CFLP) has an unbounded integrality gap, thereby implying that constant integrality ratio can not be obtained for C$k$FLP without violating the capacities even if $ k = n$. Byrka \etal ~\cite{capkmByrkaFRS2013} gave an  $O(1/\epsilon^2)$ algorithm for C$k$FLP when the capacities are uniform (UC$k$FLP)  violating the capacities by a factor of $2+\epsilon$. They use randomized rounding to bound the expected cost. It can be shown that deterministic pipage rounding cannot be used here. The strength of our techniques is demonstrated in obtaining  the first deterministic constant factor approximation with the same capacity violation. 
The primary source of inspiration for our result in Theorem~\ref{C$k$FLP-card} comes from its corollary. 
\subsection{Related Work}


Capacitated $k$-median problem has been studied extensively in the literature. For the case of uniform capacities, several results~\cite{capkmByrkaFRS2013,charikar,Charikar:1999,capkmshanfeili2014,KPR} have been obtained that violate either the capacities or the cardinality by a factor of $2$ or more. In case of non-uniform capacities, a $(7 + \epsilon)$ algorithm was given by Aardal \etal ~\cite{capkmGijswijtL2013} violating the cardinality constraint by a factor of $2$ as a special case of Capacitated $k$-FLP when the facility costs are all zero. Byrka \etal~\cite{capkmByrkaFRS2013} gave an $O(1/\epsilon)$ approximation result violating capacities by a factor of $(3+\epsilon)$.

Li~\cite{capkmshili2014} broke the barrier of $2$  in cardinality and gave an $\exp(O (1/\epsilon^2))$ approximation using at most $(1 + \epsilon) k$ facilities for uniform capacities. Li 
gave a sophisticated algorithm using a novel linear program which he calls the \emph{rectangle LP}. The result was extended to non-uniform capacities  by the same author using a new LP called \emph{configuration LP} ~\cite{Lisoda2016}. The approximation ratio was also improved from  $\exp (O(1/\epsilon^2))$ to $(O(1/\epsilon^2 \log (1/\epsilon)))$. Though the algorithm violates the cardinality only by $1 + \epsilon$, it introduces a softness bounded by a factor of $2$. The running time of the algorithm is $n^{O(1/\epsilon)}$. 

Byrka \etal ~\cite{ByrkaRybicki2015} broke the barrier of $2$ in capacities and  gave an  $O(1/\epsilon^2)$ approximation violating capacities by a factor of $(1 + \epsilon)$ factor for uniform capacities. The algorithm uses randomized rounding to round a fractional solution to the configuration LP. For non-uniform capacities, a similar result has been obtained by Demirci \etal\ ~\cite{Demirci2016}. The paper presents an $O(1/\epsilon^5)$ approximation algorithm with capacity violation by a factor of at most $(1 + \epsilon)$. The running time of the algorithm is $n^{O(1/\epsilon)}$. 

Another closely related problem to Capacitated $k$-median problem is the Capacitated $k$-center problem, where-in we have to minimize the maximum distance of a client to a facility. A 6 factor approximation algorithm was given by Khuller and Sussmann \cite{k-centerKhuller2000} for the case of uniform hard capacities (5 factor for soft capacitated case). For non-uniform hard capacities, Cygan \etal~\cite{CyganFOCS2012} gave the first constant approximation algorithm for the problem, which was further improved by An \etal~in \cite{AnMP2015} to 9 factor. 

Though the knapsack median problem (a.k.a. weighted $W$-median) is a well motivated problem and occurs naturally in practice, not much work has been done on the problem. 
Krishnaswamy \etal ~\cite{Krishnaswamysoda2011} showed that the integrality gap, for the uncapacitated case, holds even on adding the covering inequalities to strengthen the LP, and gave a $16$ factor approximation that violates the budget constraint by a factor of ($1+\epsilon$). 
Kumar~\cite{Amitsoda2012} strengthened the natural LP by obtaining a bound on the maximum distance a client can travel and gave first constant factor approximation without violating the budget constraint.
Charikar and Li ~\cite{Charikaricalp2012} reduced the large constant obtained by Kumar to $34$ which was further improved to $32$ by Swamy \cite{Swamyapprox2014}. Byrka \etal ~\cite{Byrkaesa2015} extended the work of Swamy 
and applied sparsification as a pre-processing step to obtain a factor of $17.46$. The result was further improed to $7.081(1+\epsilon)$ very recently by Krishnaswamy \etal~\cite{krishnaswamySTOC18} using iterative rounding technique, with a running time of $n^{O(1/\epsilon^2)}$.



For C$k$FLP, Aardal \etal ~\cite{capkmGijswijtL2013} extended the  FPTAS for knapsack problem to give an FPTAS for single client C$k$FLP. They also extend an $\alpha-$ approximation algorithm for (uncapacitated) $k$-median to give a $(2 \alpha + 1)-$ approximation for C$k$FLP  with uniform opening costs using at most $2k$ for non-uniform and $2k - 1$ for uniform capacities. Byrka \etal~\cite{capkmByrkaFRS2013} gave an $O(1/\epsilon^2)$ factor approximation violating the capacities by a factor of $(2 + \epsilon)$ using dependent rounding.

For CFLP, An, Singh and Svensson~\cite{Anfocs2014} gave the first LP-based constant factor approximation by strengthening the natural LP. Other LP-based algorithms known for the problem are due to Byrka \etal ~and Levi \etal ~(\cite{capkmByrkaFRS2013,LeviSS12}). The local search technique has been particularly useful to deal with capacities. The approach provides $3$ factor for uniform capacities~\cite{mathp} and $5$ factor for the non-uniform case~\cite{Bansal}.  

\subsection{Our techniques}

We extend the work of Krishnaswamy \etal~\cite{Krishnaswamysoda2011} to capacitated case. The major challenge is in writing the LP which opens sufficient number of facilities for us in bounded cost.

Filtering and clustering techniques~\cite{Lin92,Charikar:1999,LeviSS12,Shmoys,capkmByrkaFRS2013,Krishnaswamysoda2011,capkmGijswijtL2013} are used to partition the set of facilities and demands. Routing trees are used to bound the assignment costs. Main contribution of this work is a new LP and an iterative rounding algorithm to obtain a solution with at most two fractionally opened facilities.

{\bf High Level Ideas:} We first use the  filtering and clustering techniques 
to partition the set of facilities and demands. Each partition ({\em called cluster}) has sufficient opening ($\ge 1 - 1/\ell \ge 1/2$) for a fixed parameter $\ell \ge 2$ in it. An integrally open solution is obtained where-in some clusters have at least $1$ integrally opened facility and some do not have any facility opened in them. To assign the demand of the cluster that cannot be satisfied locally within the cluster,   a (directed) rooted binary routing tree is constructed, on the cluster centers. If $(s, t)$ is an edge in the routing tree then the cost of sending the unmet demand of the cluster centered at $s$ to $t$ is bounded. The edges of the tree have non-increasing costs as we go up the tree, with the root being at the top. Hence the cost of sending the unmet demand of the cluster centered at $s$ to any node $r$ up in the tree at a constant number of edges away from $s$ is bounded.
	
In order to decide which facilities to open integrally, clusters are grouped into meta-clusters of size (the number of clusters in it)  $\ell$ so as to have at least $\ell - 1$ opening in it. 
The routing tree is used to group the clusters into meta-clusters (MCs) in a top-down greedy manner, i.e., starting from the root, a meta-cluster grows by including the cluster (center) that connects to it by the cheapest edge. 
A MC grows until its size reaches $\ell$. We then proceed to make a new MC from the tree with the remaining nodes in the same greedy manner. 
This imposes a natural directed (not necessarily binary) rooted tree structure 
on the meta-clusters with the property that the edge going out of a MC is cheaper than the edges inside the MC which are further cheaper than the edges coming into the MC. Out-degree of a MC is $1$ whereas the in-degree is at most $q + 1$ where $q$ is the number of clusters in a MC.

 Next, we write a new LP to open sufficient number of facilities within each cluster and each MC. We also give an iterative rounding algorithm to solve the LP, removing the integral variables and updating the constraints accordingly in each iteration until either all the variables are fractional or all are integral. In case all the variables are fractional, we use the property of extreme point solutions to claim that the number of non-integral variables is at most two. 
Thus we obtain a solution to the LP with at most two fractional openings. 
Both the fractionally opened facilities are opened integrally at a loss of additive $f_{max}$ in the budget where $f_{max}$ is the maximum facility opening cost \footnote{Let $F'$ be the set of facilities $i$ with $f_i  > \epsilon\cdot\budget$. Enumerate all possible subsets of $F'$ of $size <= 1/\epsilon$. There are at most $n ^ {O(1/\epsilon)}$ such sets. For each such set $S$, solve the LP with  $y_i = 1\ \forall\ i \in S$ and $y_i = 0\ \forall\ i \in F' \setminus S$. The additive $f_{max}$ (which comes from the fractionally opened facilities) is $< =   \epsilon\cdot\budget$. Choose the best solution and hence theorem $1$ follows.}.

 Finally a min-cost flow problem is solved with capacities scaled up by a factor of $(2 + \epsilon)$ to obtain an integral assignment. A feasible solution to the min-cost flow problem of bounded cost is obtained as follows:
consider a scenario in which the demand accumulated within each cluster is less than $\capacity$ (we call such clusters as {\em sparse}). For the sake of easy exposition of the ideas, let each MC be of size exactly $\ell$. The LP solution opens at least $\ell - 1$ facilities integrally in each MC, with at least one facility in each cluster except for one cluster. If the cluster with unmet demand is at the root of the induced subgraph of the MC, then its demand cannot be met within the MC. We make sure that such a demand is served in the parent MC. Total demand to be served by the facilities in a MC is at most $\ell \capacity$ plus at most $(\ell + 1) \capacity$ coming from the children of the MC. Thus $(\ell- 1)$ facilities have to serve at most $(2\ell + 1)\capacity$ demand leading to a violation of $(2 + O(1/\ell))$ in capacity. Demands have to travel $O(\ell)$ edges upwards (at most $\ell$  within its own MC and at most $\ell$ in the parent MC), and hence the cost of serving them is bounded.

 The situation becomes a little tricky when there are clusters with more than $\capacity$ demand (we call such clusters as {\em dense}). One way to deal with dense clusters is to open $\floor{demand/\capacity}$ facilities integrally within such a cluster and assign the residual demand to one of them at a capacity violation of $2$. But if this cluster also has to serve $u$ units of unmet demand of one of its children (we will see later that a dense cluster has at most one child), the capacity violation could blow upto $3$ in case $\floor{demand/\capacity} = 1$. We deal with this scenario carefully.

\section{Capacitated Knapsack Median Problem}
\label{CKnM}
In this section, we consider the capacitated knapsack median problem.
CKnM can be formulated as the following integer program (IP):
\begin{eqnarray} \label{{unif-KnM}}
	Minimize & \mathcal{C}ostKnM(x,y) = \sum_{j \in \cliset}\sum_{i \in \facilityset}\dist{i}{j}x_{ij} \nonumber \\ 
	subject~ to &\sum_{i \in \facilityset}{} x_{ij} = 1 & \forall ~\singleclient \in \clientset \label{LPKMP_const1}\\ 
	& \sum_{j \in \cliset}{} x_{ij} \leq \capacity ~ y_i & \forall~ \singlefacility \in \facilityset \label{LPKMP_const4}\\ 
	& x_{ij} \leq y_i & \forall~ \singlefacility \in \facilityset , ~\singleclient \in \clientset \label{LPKMP_const2}\\   
	&\sum_{i \in \facilityset}{} f_i y_{i}\leq \budget\label{LPKMP_const3} \\
	& y_i,~x_{ij} \in \left\lbrace 0,1 \right\rbrace  \label{LPKMP_const5}
\end{eqnarray}
LP-Relaxation of the problem is obtained by allowing the variables $ y_i, x_{ij} \in [0, 1]$. Call it $LP_1$.
To begin with, we guess the facility with maximum opening cost, $f^*_{max}$, in the optimal solution and remove all the facilities with facility cost $ > f^*_{max}$ before applying the algorithm. For the easy exposition of ideas, we will give a weaker result, in section \ref{3factor}, in which we violate capacities by a factor of $3$. Most of the ideas are captured in this section.

\subsection{Simplifying the problem instance}
\label{simplify}

We first simplify the problem instance by
partitioning the sets of facilities and clients into clusters. This is achieved using the filtering technique of Lin and Vitter~\cite{Lin92}. 
For an LP solution $\rho = <x, y>$ and a subset $T$ of facilities, let $ \cardTwo{y}{T} = \sum_{i \in T}{} y_i$ denote the total extent up to which facilities are opened in $T$ under $\rho$.

{\bf Partitioning the set of facilities  into clusters and sparsifying the client set :}
Let $\rho^{*} = <x^*, y^*>$ denote the optimal $LP$ solution. Let $ \C{j} $ denote the average connection cost of a client \singleclient \ in $\rho^*$ i.e.,~$ \C{j} = \sum_{i \in \facilityset}{} x^*_{ij}\dist{i}{j}$. Let $\ell \geq 2$ be a fixed parameter  and $\ballofj{j}$ be the set of facilities within a distance of $ \ell\C{j}$ of $j~ i.e., \ \ballofj{j}  = \{ \singlefacility \in \facilityset \colon \dist{i}{j} \leq \ell\C{j} \} $ (Figure \ref{fig-clustering-a}(a)). Then, $ \cardTwo{y^*}{\ballofj{j}} \ge 1-\frac{1}{\ell}$. 
Let $\mathcal{R}_j = \ell \C{j}$ denote the \textit{radius} of $\ballofj{j}$. We identify a set $\clientset'$ of clients ( Figure~\ref{fig-clustering-a}(b)) which will serve as the centers of the clusters using Algorithm 1. 
Note that  $\ballofj{j'}  \subseteq \bundle{j'}$ and the sets $\bundle{j'}$ partition $\facilityset$. (Figure~\ref{fig-clustering-c}(b)).

\begin{algorithm} \label{algo_cc}
	\begin{algorithmic}[1]
		\STATE $\clientset' \gets \emptyset$, $S \gets \clientset$, $ctr(j)=\emptyset\ \forall j \in S$.
		\WHILE {$S \ne \emptyset$}
		\STATE Pick $j' \in S$ with the smallest radius $\mathcal{R}_{j'}$ in $S$, breaking ties arbitrarily.
		\STATE $S \gets S \setminus \{j'\}$, $\cliset' \gets \cliset' \cup \{j'\}$
		\WHILE {$\exists j \in S$: $\dist{j'}{ j} \leq 2\ell \C{j}$}
		\STATE $S \gets S \setminus \{j\}$, $ctr(j)={j'}$
		\ENDWHILE 
		\ENDWHILE 
		\STATE $\forall j' \in \cliset'$: let $ \bundle{j'} = \lbrace \singlefacility \in \facilityset \mid \forall k'\in \clientset' \colon j' \ne k'\Rightarrow \dist{i}{j'} < \dist{i}{k'} \rbrace $
	\end{algorithmic}
	\label{alg1}
	\caption{Cluster Formation}
\end{algorithm}


{\bf Partitioning the demands: }
Let $\load{i}$ denote the total demand of clients in $\clientset$ serviced by facility $i$ 
i.e., ~$\load{i} = \sum_{ j \in \clientset}{} x^{*}_{i j}$ and, $d_{j'} = \sum_{i \in \bundle{j'}}{} \load{i}$ for $j' \in \cliset'$. Move the demand $d_{j'}$ to the center $j'$ of the cluster (Figures~\ref{fig-clustering-a}-(b) and \ref{fig-clustering-c}-(a)).
For $ j \in \clientset$, let $\A{ j}{\neighbor{j'}}$ denote the total extent upto which $ j$ is served by the facilities in $\bundle{j'}$. Then, we can also write $d_{j'} = \sum_{ j \in \clientset}\A{ j}{\neighbor{j'}}$. Thus,
after this step, unit demand of any $ j \in \clientset$, is distributed to centers of all the clusters whose facilities serve $ j$.  In particular, it takes care of the demand of the clients that were removed during sparsification. Each cluster center is then responsible for the portion of demand of $ j \in \clientset$ served by the facilities in its cluster. 

\begin{figure}[t]
	\begin{tabular}{cc}
		\includegraphics[width=55mm,scale=0.5]{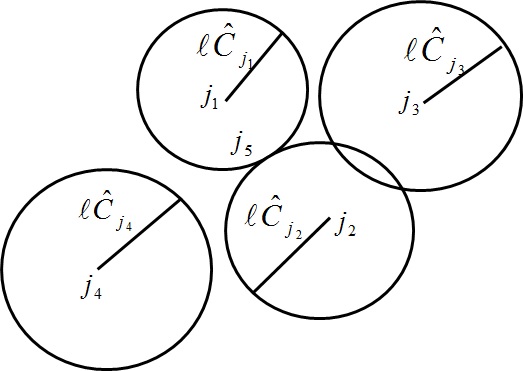}
		& \hspace{1cm}
		\includegraphics[width=55mm,scale=0.5]{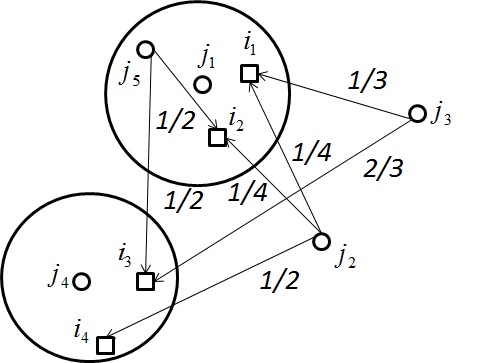}
		\\ 
		(a) & \ \  (b) 
	\end{tabular}
\caption{(a) The balls around the clients.
	(b) Reduced set of clients and assignment by LP solution.
	}
	\label{fig-clustering-a}
	\label{fig-clustering-b}
\end{figure}

\begin{figure}[t]
	\begin{tabular}{cc}
		\includegraphics[width=55mm,scale=0.5]{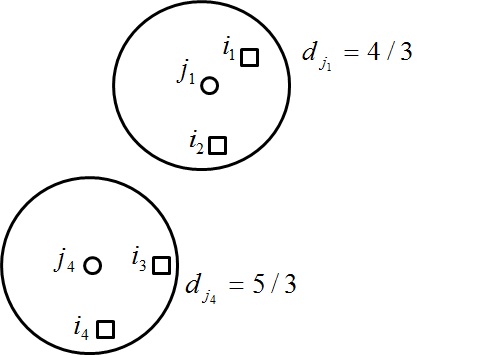}
		& \hspace{1cm}
		\includegraphics[width=55mm,scale=0.5]{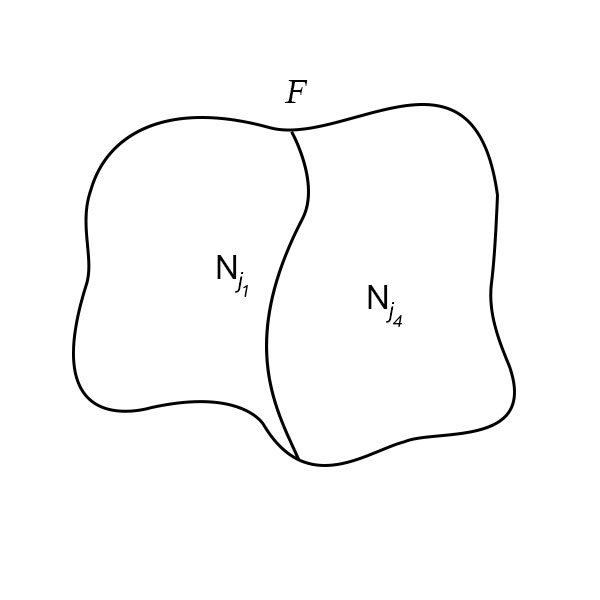}
		\\ 
		(a) & \ \ (b) 
	\end{tabular}
	\caption{(a) Partitioning of demand.
		(b) Partition of $\mathcal{F}$.
	}
	\label{fig-clustering-c}
	\label{fig-clustering-d}
\end{figure}

The cost of moving the demand $d_{j'}$ to $j'$ is bounded by $2(\ell+1)LP_{opt}$ as shown in Corollary \ref{cor2}. 
Also, any two cluster centers $j'$ and $k'$ satisfy the {\em separation property}: $\dist{j'}{k'} > 2\ell~max\{ \C{j'},\C{k'}\}$.  In addition, they satisfy Lemmas (\ref{dist-bd1}), (\ref{rad-bd}) and (\ref{eq0}). 

\begin{lemma} 
	\label{dist-bd1}
	Let $j' \in \clientset'$ and  $i \in \bundle{j'}$ then,
	($i$) For $k'\in \clientset'$, $\dist{j'}{k'}\le 2\dist{i}{ k'}$,
	($ii$) For $ j\in \cliset \setminus \clientset'$, $\dist{j'}{ j} \le 2\dist{i}{ j} + 2 \ell\C{ j}$ and ($iii$) For $ j\in \clientset$, $\dist{i}{j'} \leq \dist{i}{ j}+2\ell\C{ j}$. 
\end{lemma}
\begin{proof}

$i)$ By triangle inequality, $c(j', k') \leq c(i,j') + c(i, k')$. Since $i \in \bundle{j'} \Rightarrow c(i,j') \leq c(i,k')$ and hence $c(j', k') \leq 2 c(i, k')$.

$(ii)$ Since $ j \notin \clientset'$, there exist a client $k' \in \clientset'$ such that $ctr( j) = k'$ and $c( j,k') \leq 2\ell\C{ j}$. 
Also, If $k' = j'$ then $c(i,j') = c(i,k')$ else $c(i,j') \leq c(i,k')$ because $i \in \bundle{j'}$ and not $\bundle{k'}$.
Then, by triangle inequality, $c(i,k') \leq c(i, j) + c( j,k') \leq c(i, j) + 2\ell\C{ j} = c(i, j) + 2\mathcal{R}_{ j}$. Therefore, $c(j', j) \leq c(i,j') + c(i, j) \leq 2c(i, j) + 2\mathcal{R}_{ j}$.

$(iii)$ Consider two cases: $ j\in \clientset'$ and $ j  \notin \clientset'$. In the first case, $c(i,j') \leq c(i, j)$ because $i \in \bundle{j'}$ and not $\bundle{ j}$  and hence $c(i,j') \leq c(i, j) + 2\ell\C{ j}$. In the latter case, by triangle inequality we have, $c(i,j') \leq c(i, j) + c(j', j)$. Since $ j \notin \clientset' \Rightarrow  c(j', j) \leq 2\ell\C{ j}$. Thus, $c(i,j') \leq c(i, j) + 2\ell\C{ j}$.
\end{proof}

\begin{corollary} \label{cor2}
	$\sum_{ j \in \cliset}{}\sum_{j' \in \cliset'}{}\dist{j'}{ j}\A{ j}{\neighbor{j'}} \le 2(\ell+1)LP_{opt}$.
\end{corollary}

\begin{lemma}
	\label{rad-bd}
	Let $ j \in \clientset \setminus \clientset'$ and $j' \in \clientset'$ such that $\dist{j'}{ j} \leq \mathcal{R}_{j'}$, then $\mathcal{R}_{j'} \le 2\mathcal{R}_{ j}$.
\end{lemma}
\begin{proof} 
Suppose, if possible,  $\mathcal{R}_{j'} > 2\mathcal{R}_{ j}$. Let $ctr( j) = k' $. 
Then, $ \dist{ j}{k'} \le 2 \mathcal{R}_{ j}$. And, $\dist{k'}{j'} \le \dist{k'}{ j} + \dist{ j}{j'}$ $ \leq 2 \mathcal{R}_{ j} + \mathcal{R}_{j'} < 2\mathcal{R}_{j'} = 2\ell\C{j'}$, which is a contradiction to separation property.
\end{proof}


\begin{lemma}
	\label{eq0}
	$\sum_{j' \in {\cliset'}}{} d_{j'} 
	\sum_{i \in {\facset}}{}
	\dist{i}{j'}x^*_{ij'}
	\leq 3\sum_{j \in {\cliset}}{}
	\sum_{i \in {\facset}}{}	
	\dist{i}{j}x^*_{ij} = 
	3 LP_{opt}$.	
\end{lemma}

\begin{proof}
$\sum_{j' \in {\cliset'}}{} d_{j'} 
\sum_{i \in {\facset}}{}
\dist{i}{j'}x^*_{ij'} = \sum_{j' \in {\cliset'}}{}
\big( \sum_{ j \in \cliset}{}
\A{ j}{\neighbor{j'}}\big) \hatofC{j'}$
\\
$= \sum_{j' \in {\cliset'}}{}
\big( \sum_{ j \in \cliset : c(j',  j) \leq \R{j'}}{}
\A{ j}{\neighbor{j'}} \hatofC{j'}
+ \sum_{ j  \in \cliset : c(j',  j) > \R{j'}}{}
\A{ j}{\neighbor{j'}} \hatofC{j'}\big)
$
\\
Second term in the sum on RHS $< \frac{1}{\ell} \sum_{j' \in {\cliset'}}{}
\sum_{ j  \in \cliset : c(j',  j) > \R{j'}}{}
\A{ j}{\neighbor{j'}} \dist{j'}{ j}$
\\ 
$\leq\frac{1}{\ell} \sum_{ j \in {\cliset}}{}
\sum_{j' \in \cliset': c(j',  j) > \R{j'}}{} \sum_{i \in \bundle{j'}}{}
x_{i j}^{*}(2\dist{i}{ j} + 2\ell \hatofC{ j})
$ as $\dist{j'}{ j} \leq 2\dist{i}{ j} + 2\ell \hatofC{ j}$ by Lemma~\ref{dist-bd1}
\\
$\leq  \sum_{ j \in {\cliset}}{}
\sum_{j' \in \cliset': c(j',  j) > \R{j'}}{} \sum_{i \in \bundle{j'}}{}
x_{i j}^{*}(\dist{i}{ j} + 2 \hatofC{ j})
$. Thus the claim follows. \end{proof}


Let $\cpoor$ be the set of cluster centers $j' \in \clientset'$ for which $d_{j'} < u$ and $\crich$ be the set of remaining centers in $\clientset'$. The clusters centered at $j' \in \cpoor$ are called {\em sparse} and those centered at $j' \in \crich$  {\em dense}. 
For $j' \in \cdense$, sufficient facilities are opened in $\bundle{j'}$ so that its entire demand is served within the cluster itself and we say that $j'$ is {\em self-sufficient}. Unfortunately, the same claim cannot be made for the sparse clusters i.e., we cannot guarantee to open even one facility in each sparse cluster (since $d_{j'} < u$, we  need only one facility in each sparse cluster  $j'$). Thus, in the next section, we define a routing tree that is used to route the unmet demand of a cluster to another cluster in bounded cost.

\subsection{Constructing the Binary Routing Tree}
\label{bin-tree}
First, we define a dependency graph $G = (V, E)$, similar to the one defined by Krishnaswamy et al \cite{Krishnaswamysoda2011}, on cluster centers, i.e., $V = \clientset'$. For brevity of notation, we use $j'$ to refer to the node corresponding to cluster center $j'$ as well as to refer to the cluster center $j'$ itself. 
For  $j' \in \csparse$, let $\eta(j')$ be the nearest other cluster center in $\cliset'$ of $j'$ i.e.,~$\eta(j') = k' (\ne j') \in \cliset': k'' \in \cliset'\Rightarrow \dist{j'}{ k'} \leq \dist{j'}{ k''}$ 
and for $j' \in \cdense,~\eta(j') = j'$. The dependency graph consists of directed edges $\dist{j'}{\eta(j')}$. Each connected component of the graph is a tree except possibly for a $2$-cycle at the root.  We remove any edge 
arbitrarily from the two cycle. The resulting graph is then a forest. Note that, there is at most one dense cluster in a component and if present, it must be the root of the tree. 
The following lemma will be useful to bound the cost of sending the unserved demand of $j' \in \cpoor$ to $\eta(j')$.

\begin{lemma}	
	\label{eq3-aditya}	
	\label{nn}					
	$\sum_{j' \in { \csparse}} \bard{j'} ( 
	\sum_{i \in {\bundle{j'}}}
	\dist{i}{j'} x^*_{ij'} +
	c(j',~\eta(j')) ( 1 - \sum_{i \in \bundle{j'}}x^*_{ij'})
	)
	\leq 6 LP_{opt}		
	$.
\end{lemma}

\begin{proof}
The second term of LHS 
$=\sum_{j' \in { \csparse}} \bard{j'} \big( 
\sum_{i \notin {\bundle{j'}}}
\dist{j'}{\eta(j')} x^*_{ij'}
\big) $

$\leq\sum_{j' \in { \csparse}} \bard{j'}\big( 
\sum_{ k' \in { \cliset'}: k'\neq j'~}{} 
\sum_{i \in \bundle{ k'}}{}
\dist{j'}{k'} x^*_{ij'}
\big)$

$ \leq
\sum_{j' \in { \csparse}} \bard{j'} \big( 
\sum_{k'\in { \cliset'}: k'\neq j'~}{} 
\sum_{i \in \bundle{ k'}}{}
2 \dist{i}{j'} x^*_{ij'}
\big)$.\end{proof} 

Unfortunately, the in-degree of a node in a tree may be unbounded and hence
arbitrarily large amount of demand may accumulate at a cluster center, which may further lead to unbounded capacity violation at the facilities in its cluster.

\begin{figure}[t]
	\begin{tabular}{ccc}
		\includegraphics[width=37mm,scale=0.5]{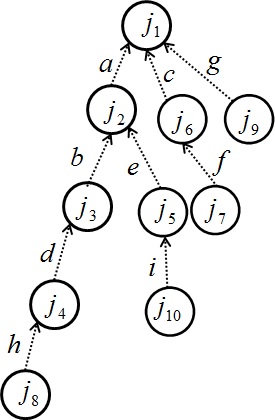}
		& 
		\includegraphics[width=40mm,scale=0.5]{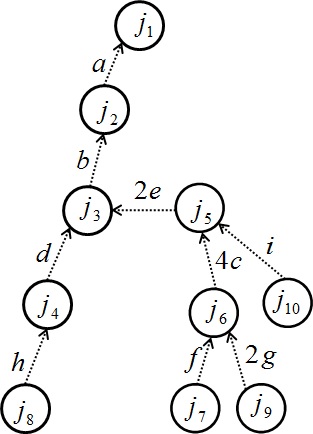}
		& 
		\includegraphics[width=50mm,scale=0.5]{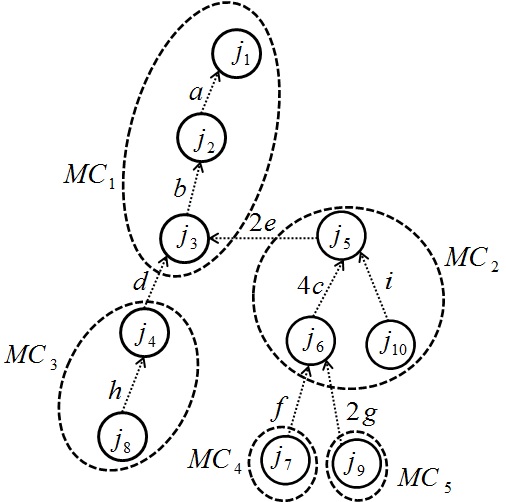} 	
		\\
		\textbf{(a)}& \ \ \textbf{(b)}& \ \ \textbf{(c)}
	\end{tabular}
	\caption{(a) A Tree $T$ of unbounded in-degree. $a<b<d<h$ , $a<c<g$ , $b<e$.
		(b) A Binary Tree $T'$ where each node has in-degree at most $2$.
		(c) Formation of meta-clusters for $\ell =3$. 
	}
	\label{bin-tree-a}
	\label{bin-tree-b}
\end{figure}

\textbf{Bounding the in-degree of a node in the dependency graph:} We convert the dependency graph $G$ into another graph $G'$ where-in the in-degree of each node is bounded by $2$ with in-degree of the root being $1$. This is done as follows (Figure~\ref{bin-tree-a}(a)-(b)): let $\mathcal{T}$ be a tree in $G$. $\mathcal{T}$ is converted into a binary tree using the standard procedure after sorting the children of node $j'$ from left to right in non-decreasing order of distance from $j'$ i.e., for each child $ k'$ (except for the nearest child) of $j'$, add an edge to its left sibling with weight $2\dist{ k'}{\eta(k')}$ and remove the edge $( k', j')$. There is no change in the outgoing edge of the leftmost child of $j'$. 
Let $\sigmaone(j')$ be the parent of node $j'$ in $G'$. Its easy to see that $\dist{j'}{\sigmaone(j')} \leq 2 \dist{j'}{\eta(j')}$. 
Henceforth whenever we refer to distances, we mean the new edge weights. Hence, we have the following: 
\begin{eqnarray}
\sum_{j' \in {\csparse}} \bard{j'} \big( 
\sum_{i \in {\bundle{j'}}}
\dist{i}{j'} x^*_{ij'} +
\dist{j'}{\sigmaone(j')} ( 1 - \sum_{i \in \bundle{j'}}x^*_{ij'} )
\big)\leq 12LP_{opt}	
\label{factor12eqn}
\end{eqnarray}

\subsection{Constructing the Meta-clusters}\label{mc}

If we could ensure that for every $j' \in \csparse$ for which no facility is opened in $\bundle{j'}$, a facility is opened in $\psi(j')$, we are done (with $3$ factor loss in capacities). But we do not know how to do that. However, for every such cluster center $j'$, we will identify a set of centers which will be able to take care of the demand of $j'$ and each one of them is within a distance of $O(\ell)\dist{j'}{\psi(j')}$ from $j'$.

We exploit the following observation to make groups of $\ell$ clusters:
each 
cluster has facilities opened in it to an extent of at least ($1 - 1/\ell$). Hence, every  collection of $\ell$ clusters, has at least $\ell - 1$ facilities opened in it. 
Thus, we make groups (called meta-clusters), each consisting of $\ell$ clusters, if possible.
For every tree $\mathcal{T}$ in $G'$, MCs are formed by processing the nodes of $\mathcal{T}$ in a top-down greedy manner starting from the root as described in Algorithm 2. (Also see Figure \ref{bin-tree-a}(c)). 
There may be some MCs of size less than $\ell$, towards the leaves of the tree.

\begin{algorithm}
	\begin{algorithmic}[1]
		\STATE \textbf{Meta-cluster(Tree $\mathcal{T}$)}\\
		\STATE $\mathcal{N} \gets set~of~nodes~in~\mathcal{T}$.
		\WHILE {there are non-grouped nodes in $\mathcal{N}$}
		\STATE Pick a topmost non-grouped node, say $k$ of $\mathcal{N}$: form a new MC, $G_k$. 
		\WHILE {$G_k$ has fewer than $\ell$ nodes}
		\STATE If $\mathcal{N} = \emptyset$ then break and stop.
		\STATE Let $j = argmin_{u \in \mathcal{N}} \{\dist{u}{v} :(u,v) \in \mathcal{T}, v \in G_k$\}, set $G_k = G_k \cup \{j\}$. $\mathcal{N} \gets \mathcal{N}\setminus \{j\}$.
		\ENDWHILE 
		\ENDWHILE 
	\end{algorithmic}
	\label{alg2}
	\caption{Meta-cluster Formation}
\end{algorithm}


 Let $G_r$ denote a MC with $r$ being the root cluster of it. With a slight abuse of notation, we will use $G_r$ to denote the collection of centers of the clusters in it as well as the set of clusters themselves. Let $\mathcal{H}(G_r)$ denote the subgraph of $\mathcal{T}$ induced by the nodes in $G_r$. $\mathcal{H}(G_r)$ is clearly a tree. 
  We say that $G_r$ is responsible for serving the demand in its clusters.

With the guarantee of only $\ell - 1$ opening amongst $\ell$ clusters, there may be a cluster with no facility opened in it. If this cluster happens to be a sparse cluster at the root, its demand cannot be served within the MC. Thus we define a (routing) tree structure on MCs as follows: a tree consists of MCs as nodes and there is an edge from a MC $G_r$ to another MC $G_s$ if there is a directed edge from root $r$ of $G_r$ to some node $s' \in \cen{s}$, $G_s$ is then called the parent meta-cluster of  $G_r$, $G_r$ a child meta-cluster of $G_s$ and the edge ($r,~s'$) is called the \textit{connecting edge} of the child MC $G_r$. If $G_r$ is a root MC, add an edge to itself with cost $\dist{r}{~\sigmaone(r)}$. This edge is then called the \textit{connecting edge} of $G_r$.  Note that the cost of any edge in $G_s$ is less than the cost of the connecting edge of $G_r$ which is further less than the cost of any edge in $G_r$. Further, a dense cluster, if present, is always the root cluster of a root MC. 
We guarantee that  the unmet demand of a MC is served in its parent MC.

\subsection{3-factor Capacity Violation}\label{3factor}

In this section, we present the main contribution of our work. Inspired by the LP of Krishnaswamy \etal~\cite{Krishnaswamysoda2011}, we formulate a new LP and present an iterative rounding algorithm to obtain a solution with at most two fractionally opened facilities. Such a solution is called {\em pseudo-integral} solution. Modifying the LP of Krishnaswamy \etal~\cite{Krishnaswamysoda2011} and obtaining a feasible solution of bounded cost for the capacitated scenario is non-trivial. The rounding algorithm is also non-trivial.

\subsubsection{Formulating the new LP and obtaining a pseudo-integral solution}
\label{newLP}

Sparse clusters have the nice property that they need to take care of small demand (< $u$ each) and dense clusters have the nice property that the total opening within each cluster is at least $1$. These properties are exploited to
	define a new LP 
	that opens sufficient number of facilities 
	in each MC such that the opened facilities are well spread out amongst the clusters (we make sure that at most 1 (sparse) cluster has no facility opened in it) and demand of a dense cluster is satisfied within the cluster itself. 
We then present an iterative rounding algorithm that provides us with a solution having at most two fractionally opened facilities.

Let $\delta_r$ be the number of dense clusters and $\sigma_r$ be the number of sparse clusters in a MC $G_r$. 
With at least $1 - 1/\ell$ opening in each sparse cluster, observing the fact that $\sigma_r \leq \ell$, we have at least $\sigma_r (1 - 1/\ell) \geq \sigma_r - 1$ total opening in $\sigma_r$ sparse clusters of $G_r$. Also, at least $\floor{\demandofj{j_d}/{\capacity}}$ opening is there in a dense cluster centered at $j_d$ in $G_r$. Let $\alpha_r = \max\{0,\sigma_r-1\}$.  LP is defined so as to open at least $\floor{\demandofj{j_d}/{\capacity}} +\alpha_r$	facilities in $G_r$.
Let $\T{j'} = \{i \in \bundle{j'}:\dist{i}{j'} \leq \dist{j'}{\sigmaone(j')}\}$ if $j' \in \csparse$ (recall that $\sigmaone(j')$ is the parent of $j'$ in binary tree) and $\T{j'} = \neighbor{j'}$ if $j' \in \cdense$. 
Also, let $\sj{r} = G_r \cap \csparse$ and $s_r = \alpha_r$ for all MCs $\cen{r}$, $\tilde{\facilityset}= \facilityset$, $\tilde{\budget} = \budget$, $r_{j'} = \floor {\demandofj{j'}/\capacity} \ \forall j' \in \cdense$ and $ \tauhat{j'} = \tau(j')\ \forall j' \in \cliset'$. These sets are updated as we go from one iteration to the next iteration in our rounding algorithm, thereby giving a new (reduced) LP in each iteration. Let $w_i$ denote whether facility $i$ is opened in the solution or not. We now write an LP, called $LP_2$ with the objective of minimising the following function: $\obj{w} =  \sumlimits{j' \in \csparse}{} ~ \demandofj{j'} [\sumap{i \in \bundle{j'}}\dist{i}{j'} w_{i} + \dist{j'}{\sigmaone(j') }  ( 1 - \sumlimits{i \in \bundle{j'}}{} w_{i})] + \capacity \sumap{j' \in \crich} \sumap{\singlefacility \in \bundle{j'}} \dist{i}{j'} w_i$ 
\begin{eqnarray}  
s.t. &\sumap{i \in \tauhat{j'}} w_{i} \leq 1 &\forall ~j' \in \csparse \label{LP-meta-clusters_const1-sparse-3}\\
&\sumap{i \in  \tauhat{j'}} w_{i} = r_{j'}&\forall ~j' \in \cdense  \label{LP-meta-clusters_const1-dense-3}\\
&\sumap{j' \in \sj{r}}~\sumap{i \in \tauhat{j'}} w_{i} \geq s_r &\forall~r:G_{r} \text{ is a MC}  \label{LP-meta-clusters_const2-3}\\
&\sumap{i \in \tilde{\facilityset}}{} f_i w_{i} \leq \tilde{\budget}  \label{LP-meta-clusters_const3-3}\\
&0 \leq w_i \leq 1 ~\forall ~i \in \tilde{\facilityset} \label{LP-meta-clusters_const5-3}
\end{eqnarray}

Constraints~(\ref{LP-meta-clusters_const1-dense-3}) and (\ref{LP-meta-clusters_const2-3}) ensure that sufficient number of facilities are opened in a meta-cluster.  Constraints~(\ref{LP-meta-clusters_const1-sparse-3})  and~(\ref{LP-meta-clusters_const1-dense-3}) ensure that the opened facilities are well spread out amongst the clusters as no more than $1$ and $\floor {\frac{\demandofj{j'}}{\capacity}}$ facilities are opened in a sparse and dense cluster respectively. Constraint~(\ref{LP-meta-clusters_const1-dense-3}) also ensures that at least $\floor {\frac{\demandofj{j'}}{\capacity}}$  facilities are opened in a dense cluster. This requirement is essential to make sure that the demand of a dense cluster is served within the cluster only. Hence, equality in constraint~(\ref{LP-meta-clusters_const1-dense-3}) is important.

\begin{lemma}
	\label{feasiblesolution-costKNM}
	A feasible solution $w'$ to $LP_2$ can be obtained such that 
	$ \obj{w'}\leq (2\ell+13) LP_{opt}$.
\end{lemma}
\begin{proof}
		Define a feasible solution to the $LP_2$ as follows: 
	let $j' \in \cdense,~i \in \T{j'}$, set $w'_i =  \frac{l_i}{d_{j'}} \floor{d_{j'}/u} = 
	\frac{l_i}{\capacity}\frac{\floor{d_{j'}/u}}{d_{j'}/u} \le \frac{l_i}{u} \leq y^*_i$. 
	For $j' \in \csparse$, we set $w'_i = min\{x^*_{ij'},~y^*_i\} = x^*_{ij'} \leq y^*_i$ for $i \in \T{j'}$ and $w'_i = 0$ for $i \in \neighbor{j'} \setminus \T{j'}$. We will next show that the solution is feasible.
	
	For $j' \in \csparse$,
	$\sumlimits{i\in \T{j'}}{}w'_i \leq \sumlimits{i\in \bundle{j'}}{}w'_{i} =  \sumlimits{i\in \bundle{j'}}{}x^*_{ij'} \leq 1$.
	
	Next, let $j' \in \cdense$, then $\sumlimits{i \in \T{j'}}{} w'_{i} = \sumlimits{i \in \neighbor{j'}}{} \frac{l_i}{\capacity}\frac{\floor{d_{j'}/u}}{d_{j'}/u} 
	= {\floor{d_{j'}/u}}$ as $\sumlimits{i \in \bundle{j'}}{} l_i = d_{j'}$.
	Note that $\sumlimits{i \in \T{j'}}{} w'_{i} \geq 1$ as $\demandofj{j'} \geq \capacity$.	
	
	For a meta-cluster $G_r$, we have 
	$\sumlimits{j' \in \cen{r}}{}~\sumlimits{i \in \T{j'}}{} w'_{i} 
	= 
	\sumlimits{j' \in \cen{r}\cap \csparse}{}~\sumlimits{i \in \T{j'}}{} x^*_{ij'} 
	\ge 
	\sumlimits{j' \in \cen{r} \cap \csparse}{} (1 - 1/l)
	= \max\{0,\sigma_r-1\} = \alpha_r$.
	
	Since for each $i \in 
	\facilityset$ we have $w'_i \leq y^*_i$  
	$\Rightarrow \sumlimits{i \in \facilityset}{} f_i w'_{i}  \leq  \sumlimits{i \in \facilityset}{} f_i y^*_{i} \leq \budget$.

	
	Next, consider the objective function. For $j' \in \crich$, we have $\sumlimits{i \in \T{j'}}{} \capacity~\dist{i}{j'} w'_i = u\sumap{i \in \bundle{j'}} \dist{i}{j'}(\frac{\sumap{j \in \clientset} x^*_{ij}}{\capacity}) = \sumap{i \in \bundle{j'}} \sumap{j \in \clientset} \dist{i}{j'} x^*_{ij} \leq \sumap{i \in \bundle{j'}} \sumap{j \in \clientset} \left( \dist{i}{j} + 2\ell \C{j} \right) x^*_{ij}$.
	Summing over all $j' \in \cdense$ we get, 
	$\sumlimits{j' \in \crich}{} \sumlimits{i \in \bundle{j'}}{} \sumlimits{j \in \clientset}{} x^*_{ij} \lbrack \dist{i}{j} +2\ell\C{j} \rbrack $
	%
	%
	%
	%
	%
	$\leq (2\ell+1) LP_{opt}$.
	
	Now consider the part of objective function for $\csparse$. $\sum_{j' \in {\csparse}}{} \demandofj{j'} ( \sum_{i \in {\neighbor{j'}}}{}\dist{i}{j'} w'_i +	\dist{j'}{\sigmaone(j')}  (1 - \sum_{i \in \bundle{j'}}{} w'_i) )$
	$= \sum_{j' \in {\csparse}}{} \demandofj{j'} ( \sum_{i \in {\T{j'}}}{}\dist{i}{j'} w'_i + \sum_{i \in {\neighbor{j'} \setminus \T{j'}}}{}\dist{i}{j'} w'_i +	\dist{j'}{\sigmaone(j')} (1 - \sum_{i \in \T{j'}}{} w'_i - \sum_{i \in \bundle{j'} \setminus \T{j'}}{} w'_i)) = \sum_{j' \in {\csparse}}{} \demandofj{j'} ( \sum_{i \in {\T{j'}}}{}\dist{i}{j'} x^*_{ij'} +	\dist{j'}{\sigmaone(j')} (1 - \sum_{i \in \T{j'}}{} x^*_{ij'}))$
	
	$< \sum_{j' \in {\csparse}}{} \demandofj{j'} ( \sum_{i \in {\T{j'}}}{}\dist{i}{j'} x^*_{ij'} +	\dist{j'}{\sigmaone(j')} (1 - \sum_{i \in \T{j'}}{}x^*_{ij'})) + 
	\sum_{j' \in {\csparse}}{} \demandofj{j'} ( \sum_{i \in {\neighbor{j'} \setminus \T{j'}}}{}(\dist{i}{j'} -	\dist{j'}{\sigmaone(j')})x^*_{ij'})$ as $\dist{i}{j'} >\dist{j'}{\sigmaone(j')}~\forall i \in \bundle{j'}\setminus\T{j'}$ 
	
	$=\sum_{j' \in {\csparse}}{} \demandofj{j'} ( \sum_{i \in {\neighbor{j'}}}{}\dist{i}{j'} x^*_{ij'} + 	\dist{j'}{\sigmaone(j')}  (1 - \sum_{i \in \bundle{j'}}{} x^*_{ij'} ) )$. Thus, by equation~(\ref{factor12eqn}), we get $\sum_{j' \in {\csparse}}{} \demandofj{j'} ( \sum_{i \in {\neighbor{j'}}}{}\dist{i}{j'} w'_i +	\dist{j'}{\sigmaone(j')}  (1 - \sum_{i \in \bundle{j'}}{} w'_i) ) \leq 12 LP_{opt}$.
	
	Thus, the solution $w'$ is feasible and $\obj{w'}$,
	\\ 
	$\sumlimits{j' \in \csparse}{} ~ \demandofj{j'} ~\left[ \sumap{i \in \bundle{j'}}\dist{i}{j'} w'_{i} + \dist{j'}{\sigmaone(j')}  \left( 1 - \sumlimits{i \in \bundle{j'}}{} w'_{i}\right) \right] + \capacity \sumap{j' \in \crich} \sumap{\singlefacility \in \bundle{j'}} \dist{i}{j'} w'_i 
	\leq (2\ell+13) LP_{opt}$.

\end{proof}

For a vector $w \in \mathcal{R}^{|\facilityset|}$ and ${\facilityset'} \subseteq \facilityset$, let $w^{\facilityset'}$ denote the vector \lq $w$ restricted to ${\facilityset'}$\rq. Also, let $\sr{} = <\sr{r}>$, $\Sr{} = <\Sr{r}>$ and $R = <r_{j'}>_{j'\in \cdense}$. Algorithm $3$ presents an iterative rounding algorithm that solves $LP_2$ and returns a pseudo-integral solution $\tilde{w}$. A sparse cluster is removed from the scenario for the next iteration as and when a facility is integrally opened in it (lines $11, 12$). In a dense cluster centered at $j'$, the number of facilities to be opened by the LP ($r_{j'}$) is decremented by the number of integrally opened facilities in it (line $15$) at every iteration and the cluster is removed when it becomes $0$ (line $16$). Similar treatment is done for $G_r \cap \csparse$ (line $12, 14$)

	\begin{algorithm}
		\begin{algorithmic}[1]
			\STATE \textbf{pseudo-integral($\tilde{\facilityset}$, $\tilde{\budget}$, $\sr{}$, $\Sr{}$, $\tauhat{}$, $R$)}\\
			\STATE $\tilde{w}^{\facilityset}_i = 0 \ \forall i \in {\facilityset}$ 
			\WHILE { $\tilde{\facilityset} \ne \emptyset$}
			\STATE Compute an extreme point solution $\tilde{w}^{\tilde{\facilityset}}$ to $LP_2$.
			\STATE $\tilde{\facilityset}_0 \gets
			\{i \in \tilde{\facilityset}:\tilde{w}^{\tilde{\facilityset}}_i = 0\}$, $\tilde{\facilityset}_1 \gets \{i \in \tilde{\facilityset}:\tilde{w}^{\tilde{\facilityset}}_i = 1\}$. 
			
			\IF {$\lvert \tilde{\facilityset}_0 \lvert = 0$ and $\lvert \tilde{\facilityset}_1 \lvert = 0$}
			\STATE Return $\tilde{w}^{\facilityset}$. 
			$\backslash*~$exit when all variables are fractionally opened$*\backslash$
			\ELSE
			\STATE For all MCs $\cen{r}$\{
			\WHILE  {$\exists j'\in \Sr{r}$ such that constraint (\ref{LP-meta-clusters_const1-sparse-3}) is tight over $\tilde{\facilityset}_1$ i.e., $\sumap{i \in \tauhat{j'} \cap \tilde{\facilityset}_1} \tilde{w}^{\tilde{\facilityset}}_{i} = 1$}
			
			\STATE Remove the constraint corresponding to $j'$ from (\ref{LP-meta-clusters_const1-sparse-3}). 
			$\backslash*~$a facility in $\tau(j')$ has been opened$*\backslash$ \nolinebreak
			\STATE set $\Sr{r} = \Sr{r} \setminus \{j'\},~\sr{r} = \max\{0,\sr{r} - 1\}$. 
			$\backslash*~$delete the contribution of $j'$ in constraint (\ref{LP-meta-clusters_const2-3})$*\backslash$
		
			\ENDWHILE	
			\STATE If $s_r=0$, remove the constraint corresponding to $\sj{r}$ from (\ref{LP-meta-clusters_const2-3}).  $\backslash*~$$\sigma_r -1$ facilities have been opened in $G_r \cap \csparse$ $*\backslash$
		
			\STATE If $ \exists j'\in G_r \cap \cdense$, set $r_{j'} \gets r_{j'} - | \tauhat{j'} \cap \tilde{\facilityset}_1|$. 
			$\backslash*$ decrement $r_{j'}$ by the number of integrally opened facilities in $\tauhat{j'}$ $*\backslash$
			
			\STATE If $r_{j'} = 0$, remove the constraint corresponding to $j'$ from (\ref{LP-meta-clusters_const1-dense-3}). $\backslash*~$$\floordjbyu{j'}$ facilities have been integrally opened in $\tau(j')$ $*\backslash$
				\}
			\ENDIF
			\STATE $\tilde{\facilityset} \gets \tilde{\facilityset} \setminus (\tilde{\facilityset}_0 \cup \tilde{\facilityset}_1)$,  
			$\tilde{\budget} \gets \tilde{\budget} - \sum_{i \in \tilde{\facilityset}_1}{}f_i\tilde{w}^{\tilde{\facilityset}}_i$,
			$\tauhat{j'} \gets \tauhat{j'} \setminus (\tilde{\facilityset}_1 \cup \tilde{\facilityset}_0) ~\forall j' \in \clientset'$.
			
			\ENDWHILE
			\STATE Return $\tilde{w}^{\facilityset}$
		\end{algorithmic}
		\label{alg3}
		\caption{Obtaining a pseudo-integral solution}
	\end{algorithm}

	\begin{lemma} 
		\label{algo-solution-costKNM-0-1}
		The solution $\tilde{w}$ given by Iterative Rounding Algorithm satisfies the following: i) $\tilde{w}$ is feasible, ii) $\tilde{w}$ has at most two fractional facilities and 
		iii) $\obj{\tilde{w}} \leq (2\ell+13) LP_{opt}$.
	\end{lemma}
	\begin{proof}
			$i)$ We will prove the claim by induction. Let $LP^{(t)}$ denote the $LP$ at the beginning of the $t^{th}$ iteration and $\tilde{w}^{(t)}$ denote the solution at the end of the $t^{th}$ iteration.
		We will show that if $\tilde{w}^{(t)}$ is a feasible solution to  $LP_2$, then $\tilde{w}^{(t+1)}$ is also a feasible solution to $LP_2$. Since $\tilde{w}^{(1)}$ is feasible (extreme point solution), the feasibility of the solution follows.
		Let $\tilde{\facilityset}^{(t)}, \tilde{\budget}^{(t)},  \sr{}^{(t)}$, $\Sr{}^{(t)}$, $\tauhat{}^{(t)}, R^{(t)}$ denote the values at the beginning of the $t^{th}$ iteration.
		Then,  
		$\tilde{w}^{(t+1)}_{i} = \tilde{w}^{(t)}_{i} ~\forall i \in \facilityset \setminus \tilde{\facilityset}^{(t+1)}$. 

		
		Consider a constraint that was not present in $LP^{(t+1)}$. In any iteration, we remove a constraint only when none of the facilities in its corresponding clusters is fractionally opened. That is all the facilities in $\tau({j'})$ appearing on the left hand side of a constraint are integral. Thus $\tilde{w}^{(t+1)}_{i} = \tilde{w}^{(t)}_{i}$ for all such facilities. Hence if they are satisfied by $\tilde{w}^{(t)}$ then they are satisfied by $\tilde{w}^{(t + 1)}$. So, we consider only those constraints that were present in $LP^{(t+1)}$. For $j' \in \csparse$, since $\tauhat{j'}^{(t+1)} = \tau(j') \setminus \tilde{\facilityset}_{0}^{(t)}~\forall t$, therefore, $\sumap{i \in \tauhat{j'}^{(t+1)}} \tilde{w}_i^{(t+1)} = \sumap{i \in \tau(j')} \tilde{w}_i^{(t+1)}~\forall t$. Thus, we will omit $(t)$ and use $\tau()$ instead of $\tauhat{}$ for brevity of notation. 
		
		Consider constraints~(\ref{LP-meta-clusters_const1-sparse-3}) that were not removed in $t^{th}$ iteration. Since $\tau(j') \subseteq \tilde{\facilityset}^{(t+1)}$ for $j' \in \csparse$, the feasibility of the constraint follows as $\tilde{w}^{(t+1)}$ is an extreme point solution of the reduced $LP$ over the set $\tilde{\facilityset}^{(t+1)}$.	
		
		Next, consider constraints~(\ref{LP-meta-clusters_const1-dense-3}). Let $\facilityset_1^{(t)}$ denote the set of facilities that are opened integrally in $\tilde{w}^{(t)}$ i.e., $\tilde{w}^{(t)}_{i} = 1 ~\forall i \in \facilityset_1^{(t)}$ then
		the corresponding constraint in $LP^{(t+1)}$ is 
		$\sumap{i \in \tau({j'}) \setminus \facilityset_1^{(t)}} w_i = \floor{\frac{d_{j'}}{u}} - |\facilityset_1^{(t)}|$. Since $\tilde{w}^{(t +1)}$ is an extreme point solution of $LP^{(t+1)}$, it satisfies this constraint i.e., $\sumap{i \in \tauhat{j'} \setminus \facilityset_1^{(t)}} \tilde{w}^{(t+1)}_{i} = \floor{\frac{d_{j'}}{u}} - |\facilityset_1^{(t)}|$. Since $w^{(t+1)}_{i} = w^{(t)}_{i} = 1  ~\forall i \in \facilityset_1^{(t)}$, adding $\facilityset_1^{(t)}$ on both the sides, we get the desired feasibility.

		Consider constraints (\ref{LP-meta-clusters_const2-3}). Since $\tilde{w}^{(t)}$ is feasible for $LP_2$, we have, $\sumap{j' \in G_r \cap \csparse} \sumap{i \in \tau(j')} \tilde{w}^{(t)}_{i} \geq \alpha_r$ and since $\tilde{w}^{(t+1)}$ is feasible for $LP^{(t+1)}$, we have $\sumap{j' \in \Sr{r}^{(t+1)}} \sumap{i \in \tau{(j')}} \tilde{w}^{(t+1)}_{i} \geq \sr{r}^{(t+1)}$. 
		Then, $\sumap{j' \in G_r \cap \csparse} \sumap{i \in \tau(j')} \tilde{w}^{(t+1)}_{i} = \sumap{j' \in (G_r \cap \csparse) \setminus \Sr{r}^{(t+1)}} \sumap{i \in \tau(j')} \tilde{w}^{(t+1)}_{i} + \sumap{j' \in \Sr{r}^{(t+1)}} \sumap{i \in \tau(j')} \tilde{w}^{(t+1)}_{i}$
		$ \ge \sumap{j' \in (G_r \cap \csparse) \setminus \Sr{r}^{(t+1)}} \sumap{i \in \tau(j')} \tilde{w}^{(t)}_{i} + \sr{r}^{(t+1)}$
		$ = \sumap{j' \in (G_r \cap \csparse) \setminus \Sr{r}^{(t+1)}} 1 + \sr{r}^{(t+1)}$ (as these clusters must have been removed as they got tight)
		$= |(G_r \cap \csparse) \setminus \Sr{r}^{(t+1)}| + \sr{r}^{(t+1)}$ = $\alpha_r$

		Next, consider constraint (\ref{LP-meta-clusters_const3-3}).	Since $\tilde{w}^{(t)}$ is feasible for $LP_2$, we have  $\sumap{i \in\facilityset} f_i\tilde{w}^{(t)}_{i} \leq \budget$  and since $\tilde{w}^{(t+1)}$ is feasible for $LP^{(t+1)}$, we have $\sumap{i \in \tilde{\facilityset}^{(t+1)}} f_i\tilde{w}^{(t+1)}_{i} \leq \tilde{\budget}^{(t+1)}$. Also, we have $w^{(t+1)}_{i} = w^{(t)}_{i}  ~\forall i \in {\facilityset} \setminus \tilde{\facilityset}^{(t+1)}$. Consider $\sumap{i \in \facilityset} f_i\tilde{w}^{(t+1)}_{i} = \sumap{i \in \facilityset \setminus \tilde{\facilityset}^{(t+1)}}  f_i\tilde{w}^{(t+1)}_{i} + \sumap{i \in \tilde{\facilityset}^{(t+1)}} f_i\tilde{w}^{(t+1)}_{i} \leq \sumap{i \in \facilityset \setminus \tilde{\facilityset}^{(t+1)}}  f_i\tilde{w}^{(t)}_{i} + \tilde{\budget}^{(t+1)} $. And since $\tilde{\budget}^{(t+1)} = \budget - \sumap{i \in \facilityset \setminus \tilde{\facilityset}^{(t+1)}}  f_i\tilde{w}^{(t)}_{i}$, we have $\sumap{i \in \facilityset} f_i\tilde{w}^{(t+1)}_{i}  \leq \budget$. Thus, the solution $\tilde{w}^{(t+1)}$ is feasible.
	
%
	$ii)$	Consider the last iteration of the algorithm. The iteration ends either at step $(3-4)$ or at step $(9-10)$. In the former case, the solution clearly has no fractionally opened facility. Suppose we are in the latter case. Let the linearly independent tight constraints corresponding to (\ref{LP-meta-clusters_const1-sparse-3}),
		(\ref{LP-meta-clusters_const1-dense-3}) and
		(\ref{LP-meta-clusters_const2-3}) be denoted as $\mathcal{X}$, $\mathcal{Y}$ and $\mathcal{Z}$ respectively. 
		Let $A$ and $B$ be set of variables corresponding to some constraint in $\mathcal{X}$ and $\mathcal{Z}$ respectively such that $A\ \cap B \neq \emptyset$. 
		Then, $A \subseteq B$. Imagine deleting $A$ from $B$ and subtracting $1$ 
		from $\sr{r}$. Repeat the process with another such constraint in $\mathcal{X}$ until there is no more constraint in $\mathcal{X}$ whose variable set has a non-empty intersection with $B$. At this point, $\sr{r} \geq 1$ 
		and the number of variables in $B$ is at least $2$. Number of variables in any set corresponding to a tight constraint in $\mathcal{X}$ (or $\mathcal{Y}$) is also at least $2$. Thus, the total number of variables is at least $2 |\mathcal{X} | + 2 |\mathcal{Y}| + 2 |\mathcal{Z}|$ and the number of tight constraints is at most  $ |\mathcal{X} | + |\mathcal{Y}| + |\mathcal{Z}| + 1$. 
		Thus, we get $|\mathcal{X}| + |\mathcal{Y}| + |\mathcal{Z}| \leq 1$ and hence there at most two (fractional) variables.
		
	$iii)$ Note that no facility is opened in $\bundle{j'}\setminus \T{j'}:j' \in \csparse$ for if $i \in \bundle{j'}\setminus\T{j'}:j'\in \csparse$ is opened, then it can be shut down and the demand $d_{j'} \tilde{w}_i$,  can be shipped to $\sigmaone(j')$, decreasing the cost as $\dist{j'}{\sigmaone(j')} < \dist{i}{j'}$.
		Then, the claim follows as we compute extreme point solution in step (7) in the first iteration and the cost never increases in subsequent calls.

	\end{proof}

\subsubsection{Obtaining an integrally open solution} \label{sectionfinal}

The two fractionally opened facilities obtained in Section~\ref{newLP}, if any, are opened integrally at a loss of additive $f_{max}$ in the budget. Let $\hat{w}$ denote the 
 solution obtained.
 Next lemma shows that $\hat{w}$ has sufficient number of  facilities opened in each MC to serve the demand the MC is responsible for, except possibly for $u$ units. Lemma (\ref{io-in-grp-knm_new1.1}) presents the assignments done within a MC and discusses their impact on the capacity and the cost bounds.

\begin{lemma}
	\label{io-in-grp-knm_new1.1}
	Consider a meta-cluster $G_r$. Suppose the capacities are scaled up by a factor of $\max\{3, 2+\frac{4}{\ell-1}\}$ for $\ell \geq 2$. Then, i) the dense cluster in $G_r$ (if any) is {\em self-sufficient} i.e., its demand can be completely assigned within the cluster itself at a loss of at most factor 2 in cost. ii) There is at most one cluster with no facility opened in it and it is a sparse cluster. iii) Any (cluster) center responsible for the unserved demand of $j' \in \cliset'$ is an ancestor of $j'$ in $\mathcal{H}(G_r)$. iv) At most $u$ units of demand in $G_r$ remain un-assigned and it must be in the root cluster of $G_r$. Such a MC cannot be a root MC. v) Let $\beta_r =  \floordjbyu{\dense} + \max \{0, \sigma_r - 1\} $, where $\dense$ is the center of the dense root cluster (if any) in $\cen{r}$. Then, at least $\beta_r$ facilities are opened in $\cen{r}$. (vi) Total distance traveled by demand $\demandofj{j'}$ of $j' (\ne r) \in G_r $ to reach the centers of the clusters in which they are served is bounded by $d_{j'}\dist{j'}{\sigmaone(j')}$. 
\end{lemma}

\begin{proof}
	($i$) Let $j_d \in \cdense \cap \cen{r}$. Total demand $d_{j_d}$ of $j_d$ can be distributed to the opened facilities ($\geq \floordjbyu{\dense}$) at a loss of factor $2$ in capacity and cost both, as $d_{j_d}/u - \floordjbyu{\dense} < 1 \leq \floordjbyu{\dense}$.
	
	
	For $\sigma_r=0$, ($ii$) - ($v$) hold vacuously. So, let $\sigma_r \geq 1$ ($ii$) $LP_2$ opens $\alpha_r = \max \{0, \sigma_r - 1\} $ facilities in $\cen{r} \cap \csparse$. Constraint (\ref{LP-meta-clusters_const1-sparse-3}) ensures that at most one facility is opened in each sparse cluster. Thus, there is at most one cluster in $\cen{r} \cap \csparse$ with no facility opened in it.  
	($iii$) \& ($iv$) Let $j' \in \cen{r}\cap\csparse$ such that no facility is opened  in $\T{j'}$. If $j'$ is not the root of $G_r$ or $G_r$ is a root MC, then 
	$LP_2$ must have opened a facility in $\T{\sigmaone(j')}$. Demand of $j'$ is assigned to this facility at a loss of maximum $2$ factor in capacity if $\sigmaone(j') \in \csparse$ and 3 if $\sigmaone(j') \in \cdense$: $d_{\psi(j')} = 1.99u$ and $d_{j'} = .99u$.
	Otherwise (if $j'$ is the root of $G_r$ and $G_r$ is not a root MC), at most $u$ units of demand of $\cen{r}$ remain unassigned within $\cen{r}$. ($v$) holds as$ \floordjbyu{\dense}$ facilities are opened in the cluster centered at $j_d$ and $\alpha_r=\max \{0, \sigma_r - 1\} $ facilities are opened in $G_r \cap \csparse$ by constraints (\ref{LP-meta-clusters_const1-dense-3}) and (\ref{LP-meta-clusters_const2-3}) respectively.
	($vi$) Since the demand  $\demandofj{j'}$ of $j' \in G_r $ is served either within its own cluster or in the cluster centered at $\sigmaone(j')$, total distance traveled by demand $\demandofj{j'}$ of $j'$ to reach the centers of the clusters in which they are served is bounded by $d_{j'}\dist{j'}{\sigmaone(j')}$. 	
\end{proof}

Lemma (\ref{io-in-grp-knm_new2.1}) deals with the remaining demand that we fail to assign  within a MC. Such demand is assigned in the parent MC. Lemma (\ref{io-in-grp-knm_new2.1}) discusses the cost bound for such assignments and the impact of the demand coming onto $\cen{r}$ from the children MCs along with the demand within $\cen{r}$ on capacity.

\begin{lemma}
	\label{io-in-grp-knm_new2.1}
	Consider a meta-cluster $G_r$. The demand of $G_r$ and the demand coming onto $G_r$ from the children meta-clusters can be assigned to the facilities opened in $G_r$ such that: i) capacities are violated at most by a factor of  $max\{3, 2+\frac{4}{\ell-1}\}$ for $\ell \geq 2$. ii) Total distance traveled by demand $\demandofj{j'}$ of $j' \in \cliset'$ to reach the centers of the clusters in which they are served is bounded by $\ell d_{j'}\dist{j'}{\sigmaone(j')}$.
\end{lemma}

\begin{proof}
	After assigning the demands of the clusters within $G_r$ as explained in Lemma (\ref{io-in-grp-knm_new1.1}), demand coming from all the children meta-clusters are distributed proportionately to facilities within $G_r$ utilizing the remaining capacities. Next, we will show that this can be done within the claimed capacity bound.
	

	($i$) Let $G_r$ be a non leaf meta-cluster with a dense cluster $j' \in \cdense$ at the root, if any. Also, let $t_r$ be the total number of clusters in $G_r$, i.e., $t_r= \delta_r+\sigma_r$.
	The total demand to be served in $G_r$ is at most $u(\floor{d_{j'}/u} + 1 + \sigma_r) + u (t_r+1) \leq
	(\beta_r + 2)u + (t_r+1) u$ whereas the total available capacity is at least $\beta_r u$ by Lemma (\ref{io-in-grp-knm_new1.1}). Thus, the capacity violation is bounded by
	$\frac{(\beta_r + 2)u + (t_r+1) u}{\beta_r u} \leq \frac{(\beta_r + 2)u + (\beta_r+2) u}{\beta_r u} = 2+4/\beta_r \leq 2+4/(\ell-1)$ (as $\floor{d_{j'}/u} \geq \delta_r$ we have  $\beta_r \geq \sigma_r - 1 + \delta_r = t_r - 1 =  \ell-1$ for a non-leaf MC).
	
	The capacity violation of factor $3$ can happen in the case when no facility is opened in $\T{j'}$ for $j' \in \csparse$ and $\sigmaone(j') \in \cdense$ as explained in Lemma (\ref{io-in-grp-knm_new1.1}).
	
	Leaf meta-clusters may have length less than $l$ but they do not have any demand coming onto them from the children meta-cluster, thus capacity violation is bounded as explained in Lemma (\ref{io-in-grp-knm_new1.1}). 

	($ii$) Let $j'$ belongs to a MC $G_r$ such that its demand is not served within $G_r$. Then, $j'$ must be the root of $G_r$ and its demand is served by facilities in clusters of the parent MC, say $G_s$. Since the edges in $G_s$ are no costlier than the connecting edge $(j', \sigmaone(j'))$ of $G_r$ and there are at most $\ell-1$ edges in $G_s$, the total distance traveled by demand $\demandofj{j'}$ of  $j'$ to reach the centers of the clusters in which they are served is bounded by $\ell d_{j'}\dist{j'}{\sigmaone(j')}$.\end{proof}

Choosing $\ell \geq 2$ such that $2 + \frac{4}{(\ell-1)} = 3 \Rightarrow \ell = 5$. Lemma~(\ref{bd-cost-to-facility}) bounds the cost of assigning the demands collected at the centers to the facilities  opened in their respective clusters. 
\begin{lemma}
\label{bd-cost-to-facility}
	The cost of assigning the demands collected at the centers to the facilities opened in their respective clusters is bounded by $O(1)LP_{opt}$.
\end{lemma}
\begin{proof}
Let $j' \in \cliset'$. Let $\lambda(j')$ be the set of centers $j''$ such that facilities in $\T{j''}$ serve the demand of $j'$. Note that if some facility is opened in $\T{j'}$, then $\lambda(j')$ is $\{j'\}$ itself and if no facility is opened in $\T{j'}$, then $\lambda(j') = \{j'': \exists i \in \T{j''}$ \text{ such that demand of} $j'$ \text{is served by} $i$ as per the assignments done in Lemmas (\ref{io-in-grp-knm_new1.1}) and (\ref{io-in-grp-knm_new2.1})$\}$. 

The cost of assigning a part of the demand $d_{j'}$ to a facility opened in $\lambda(j') \cap \csparse$ is bounded differently from the part assigned to facilities in $\lambda(j') \cap \cdense$.

Let $j'' \in \csparse \cap \lambda(j'),~i \in \T{j''}$. Then, $\dist{j''}{i} \leq \dist{j''}{\sigmaone(j'')} \leq \dist{j'}{\sigmaone(j')}$.  This was the motivation to define $\tau(j')$ the way it was, while defining $LP_2$.  Last inequality follows as: either $j''$ is above $j'$ in the same MC (say $G_r$) (by Lemma (\ref{io-in-grp-knm_new1.1}.3)) or $j''$ is in the parent MC (say $G_s$) of $G_r$. 
In the first case, the edge ($j'',~\sigmaone(j'')$) is either in $G_r$ or is the connecting edge of $G_r$. The inequality follows as edge costs are non-increasing as we go up the tree. In the latter case, edge ($j'',~\sigmaone(j'')$) is either in $G_s$ or it is the connecting edge of $G_s$: in either case, $\dist{j''}{\sigmaone(j'')} \leq \dist{j'}{\sigmaone(j')}$ as the connecting edge of $G_s$ is no costlier than the edges in $G_s$ which are no costlier than the connecting edge of $G_r$ (possibly $\dist{j'}{\sigmaone(j')}$) which are no costlier than the edges in $G_r$.
Summing over all $j', j'' \in \csparse$, 
we see that this cost is bounded by $O(1)LP_{opt}$.	

Next, let $j'' \in \cdense \cap \lambda(j'), i \in \bundle{j''}$. Further, let ${g}_{i}$ be the total demand served by a facility $i$. 
Since ${g}_{i} \leq 3u$, the cost of transporting $3u$ units of demand from $j''$ to $i$ is $3u\hat{w}_i\dist{i}{j''}$. Summing it over all $i \in \bundle{j''}$, $j'' \in \cdense$, and then over all $j' \in \cliset'$, 
we get that the total cost for $\cdense$ is bounded by $O(1)LP_{opt}$. 
\end{proof} 

\subsection{$(2+\epsilon)$ factor capacity violation}

There is only one scenario in which we violate the capacities by a factor of $3$ in the previous section. In all other scenarios capacities scaled up by a factor of $(2 + \epsilon)$ are sufficient even to accommodate the demand of the children MCs. Consider this special scenario. Let $j_d$ be the center of the dense cluster and $j_s$ be its only child (sparse) cluster in the routing tree. Further let, $d_{j_d} = 1.99u$ and $d_{j_s} = .99u$. Then, we must have a total opening of more than $2$ in the clusters of $j_d$ and $j_s$ taken together whereas $LP_2$ opens only $1$. In such a scenario, if we treat $j_s$ with $j_d$ instead of considering it with the remaining sparse clusters of $G_r$, we can open $2$ facilities in $\tau({j_d}) \cup \tau({j_s})$ and they have to serve a total demand of at most $4\capacity$ ($1.99u + .99u  +$ at most $u$ of the remaining sparse clusters) within the MC, thereby violating the capacities by a factor of at most $2$. On the other hand, if $d_{j_d} = 1.01u$ and $d_{j_s} = .98u$, then we cannot guarantee to open $2$ facilities in $\tau({j_d}) \cup \tau({j_s})$. In this case, if we treated $j_s$ with $j_d$ and only $1$ facility is opened in $\tau({j_d}) \cup \tau({j_s})$, it will have to serve  a total demand of (close to) $3\capacity$ ($1.01u + .98u  +$ at most $u$ of the remaining sparse clusters) leading to violation of $3$ in capacity. 
Note that first case corresponds to the scenario when the residual demand of $j_d$ (viz. $.99u$ here) is large (close to $\capacity$) and the second case corresponds to the scenario when the residual demand of $j_d$ (viz. $.01u$ here) is small (close to $0$). In the first case we treat $j_s$ with $j_d$ whereas in the second case, we treat it with the remaining sparse clusters. 
In Section \ref{3factor}, one can imagine that a MC $G_r$ is partitioned into $\mcone{r}$ and $\mctwo{r}$ where $\mcone{r}$ contained only the dense cluster of $G_r$ and $\mctwo{r}$ contained all the sparse clusters of $G_r$.
We modify the partitions as follows: let $\resj{\dense} =  d_{j_d}/u - \floordjbyu{j_d}$: $(i)$ if $\resj{\dense}< \epsilon$: set $\mcone{r} = G_r \cap \cdense$, $\mctwo{r} = G_r \cap \csparse$, $\gamma_r=\floordjbyu{\dense}$, $ \sigma'_r = \sigma_r$. (This is same as above.) $(ii)$ otherwise, $ \epsilon \leq \resj{\dense}< 1$: set $\mcone{r} = (G_r \cap \cdense) \cup \{j_s\} $, $\mctwo{r} = (G_r \cap \csparse) \setminus \{j_s\}$, $\gamma_r=\floordjbyu{\dense} + \lvert \{j_s\} \lvert$ \footnote{ In case a component of dependency graph consists of a singleton dense cluster, $j_s$ may not exist. This case causes no problem even if $res(j_d)$ is large as it must be a leaf MC in this case.}, $ \sigma'_r = \max\{0, \sigma_r - 1\}$. 


We modify our $LP$ accordingly
so as to open at least $\gamma_r$ facilities in $\mcone{r}$ and $\alpha_r = \max\{0, \sigma'_r - 1\}$
facilities in $\mctwo{r}$.
Let $\Soner{r}=\mcone{r},~\soner{r}=\gamma_r$ and $\Stwor{r}=\mctwo{r},~\stwor{r}=\alpha_r, \tauhat{j'} = \T{j'}~\forall j'$. For $j' \in \cdense$, let $r_{j'} = \floor {\demandofj{j'}/\capacity}$.
Also, let $\tilde{\facilityset}= {\facilityset}$ and $\tilde{\budget}= {\budget}$. Let $w_i$ denote whether facility $i$ is opened in the solution or not. $LP_2$ is modified as follows: 

\noindent $LP_3:$ \textit{Min.}~$\obj{w}$ 
\begin{eqnarray}  
	&subject~ to \sum_{i \in \tauhat{j'}}{} w_{i} \leq 1 &\forall ~j' \in \csparse \label{LP-meta-clusters_const1-sparse}\\
	&\sumap{j' \in \Soner{r}} \sumap{i \in \tauhat{j'}} w_{i} \geq \soner{r} &\forall~\mcone{r}: \soner{r} \ne 0 \label{LP-meta-clusters_const1-group-CKM}\\
	&\sumap{j' \in \Stwor{r}}~\sumap{i \in \tauhat{j'}} w_{i} \geq \stwor{r}  &\forall~\mctwo{r}: \stwor{r} \ne 0 \label{LP-meta-clusters_const2}\\
	&\sum_{i \in \tilde{\facilityset}} f_iw_i \leq \tilde{\budget} \label{LP-meta-clusters_const3}\\
	&0 \leq w_i \leq 1 &\forall i \in \tilde{\facilityset}\label{LP-meta-clusters_const5}
\end{eqnarray}



\begin{lemma}
	A feasible solution $w'$ to $LP_3$ can be obtained such that 
	$ \obj{w'}\leq (2\ell+13) LP_{opt}$.
\end{lemma}
\begin{proof}
	Proof is similar to the proof of Lemma (\ref{feasiblesolution-costKNM}).
\end{proof}
Algorithm 3 can be modified to obtain Algorithm 4 as follows: whenever a constraint corresponding to (\ref{LP-meta-clusters_const1-sparse}) gets tight over integrally opened facilities, it is removed from $\Soner{r}$ or $\Stwor{r}$ wherever it belongs, in the same manner as line $12$ of Algorithm 3. 

\begin{algorithm}
	\begin{algorithmic}[1]
		\STATE \textbf{pseudo-integral($\tilde{\facilityset}$, $\tilde{\budget}$, $\soner{}$, $\stwor{}$, $\Soner{}$, $\Stwor{}$, $\tauhat{}$, $R'$	)}\\
		\STATE $\tilde{w}^{\facilityset}_i = 0 \ \forall i \in {\facilityset}$ 
		\WHILE { $\tilde{\facilityset} \ne \emptyset$}
		\STATE Compute an extreme point solution $\tilde{w}^{\tilde{\facilityset}}$ to $LP_3$.
		\STATE $\tilde{\facilityset}_0 \gets
		\{i \in \tilde{\facilityset}:\tilde{w}^{\tilde{\facilityset}}_i = 0\}$, $\tilde{\facilityset}_1 \gets \{i \in \tilde{\facilityset}:\tilde{w}^{\tilde{\facilityset}}_i = 1\}$. 
		
		\IF {$\lvert \tilde{\facilityset}_0 \lvert = 0$ and $\lvert \tilde{\facilityset}_1 \lvert = 0$}
		\STATE Return $\tilde{w}^{\facilityset}$.
		\ELSE
		\STATE For all MCs $\cen{r}$\{
		
		\WHILE  {$\exists j'\in \cen{r} \cap \cpoor$ such that constraint (\ref{LP-meta-clusters_const1-sparse}) is tight over $\tilde{\facilityset}_1$ i.e., $\sumap{i \in \tauhat{j'} \cap \tilde{\facilityset}_1} \tilde{w}^{\tilde{\facilityset}}_{i} = 1$}
		
		\STATE Remove the constraint corresponding to $j'$ from~(\ref{LP-meta-clusters_const1-sparse}). $\backslash*~$a facility in $\tau(j')$ has been opened$*\backslash$
		
		\STATE If $j'\in \Soner{r}$, set $\Soner{r} = \Soner{r} \setminus \{j'\},~\soner{r} = \max\{0,\soner{r} - 1\}$. $\backslash*~$delete the contribution of $j'$ in constraint (\ref{LP-meta-clusters_const1-group-CKM}) $*\backslash$
		
		
		\STATE If $j'\in \Stwor{r}$, set $\Stwor{r} = \Stwor{r} \setminus \{j'\},~\stwor{r} =\max\{0, \stwor{r} - 1\}$.$\backslash*~$delete the contribution of $j'$ in constraint (\ref{LP-meta-clusters_const2}) $*\backslash$ 
		
		\STATE If $\stwor{r} = 0$, remove the constraint corresponding to the MC from~(\ref{LP-meta-clusters_const2}).$\backslash*~$$\alpha_r$ facilities have been opened in $G_r \cap \csparse$ $*\backslash$
		\ENDWHILE
		
		\STATE If $\exists j'\in \cen{r} \cap \cdense$, set $\soner{r} = \soner{r} - | \tauhat{j'} \cap \tilde{\facilityset}_1|$. $\backslash*$ decrement $\soner{r}$ by the number of integrally opened facilities in $\tauhat{j'}$ $*\backslash$
		
		\STATE If $\soner{r} = 0$, remove the constraint corresponding to the MC from~(\ref{LP-meta-clusters_const1-group-CKM}). $\backslash*~$$\gamma_{r}$ facilities have been opened in $\mcone{r}$ $*\backslash$
		
		\ENDIF
		\STATE $\tilde{\facilityset} \gets \tilde{\facilityset} \setminus (\tilde{\facilityset}_0 \cup \tilde{\facilityset}_1)$,  
		$\tilde{\budget} \gets {\tilde{\budget}} - \sum_{i \in \tilde{\facilityset}_1}{}f_i\tilde{w}^{\tilde{\facilityset}}_i$,
		$\tauhat{j'} \gets \tauhat{j'} \setminus (\tilde{\facilityset}_1 \cup \tilde{\facilityset}_0) ~\forall j' \in \clientset'$.
		
		\ENDWHILE
		\STATE Return $\tilde{w}^{\facilityset}$.
	\end{algorithmic}
	\label{alg4}
	\caption{Obtaining a pseudo-integral solution}
\end{algorithm} 

\begin{lemma} 
	The solution $\tilde{w}$ given by Iterative Rounding Algorithm satisfies the following: i) $\tilde{w}$ is feasible, ii) $\tilde{w}$ has at most two fractional facilities and 
	iii) $\obj{\tilde{w}} \leq (2\ell+13) LP_{opt}$.
\end{lemma}

\begin{proof}
	Proof is similar to the proof of Lemma (\ref{algo-solution-costKNM-0-1}).
\end{proof}

%
The two fractionally opened facilities, if any, are opened integrally as in Section \ref{sectionfinal} at a loss of additive $f_{max}$ in the budget. Let $\hat{w}$ denote the integrally open solution.

In the next lemma, we show that $\hat{w}$ has sufficient number of  facilities opened in each MC to serve the demand the MC is responsible for, except possibly for $u$ units.  
Let $M$ be the set of all meta clusters and $M_1$ be the set of meta clusters, each consisting of exactly one dense and one sparse cluster. MCs in $M_1$ need special treatment and will be considered separately. Lemma (\ref{io-in-grp-knm_new1}) presents the assignments done within a MC and discusses their impact on the capacity and the cost bounds. 

\begin{lemma}
	\label{io-in-grp-knm_new1}
	Consider a meta-cluster $\cen{r}$. Suppose the capacities are scaled up by a factor of $ 2+\epsilon$ for $\ell \geq 1/\epsilon$. Then,  
	(i) $\mcone{r}$ is {\em self-sufficient} i.e.,~its demand can be completely assigned within the cluster itself.	
	(ii) There are at most two clusters, one in $\mcone{r}$ and one in $\mctwo{r}$, with no facility opened in them and these clusters are sparse.
	(iii) Any (cluster) center responsible for the unserved demand of $j'$ is an ancestor of $j'$ in $\mathcal{H}(G_r)$. 
	(iv) At most $u$ units of demand in $\cen{r}$ remain un-assigned and it must be in the root cluster of $\cen{r}$. Such a MC cannot be a root MC. 
	(v) For $\cen{r} \in M \setminus M_1$, let $\beta_r =  \floordjbyu{\dense} + \max \{0, \sigma_r - 1\} $, where $\dense$ is the center of the dense root cluster in $\cen{r}$. Then, at least $\beta_r$ facilities are opened in $\cen{r}$. 
	(vi) 
	For $\cen{r} \in  M_1$, let $\beta_r = \floordjbyu{\dense}$ if $\resj{\dense} < \epsilon$  and $= \floordjbyu{\dense} + 1$ otherwise. Then, at least $\beta_r$ facilities are opened in $\cen{r}$.
	(vii) Total distance traveled by demand $\demandofj{j'}$ of $j' (\ne r) \in \cen{r} $ to reach the centers of the clusters in which they are served is bounded by $2d_{j'}\dist{j'}{\sigmaone(j')}$. 
\end{lemma}
\begin{proof}
		($i$) Let $j_d \in \cdense \cap \mcone{r}$.
		Consider the case when $res(j_d) < \epsilon$. The total demand $(\floordjbyu{\dense} + res(j_d))u  \leq (\floordjbyu{\dense} + \epsilon)u$ of $\mcone{r}$ can be distributed to the  opened facilities ($\geq \floordjbyu{\dense}$) at a loss of factor $2$ in capacity as $ \floordjbyu{\dense} \geq 1$.
		
		When  $\epsilon \leq res(j_d) <1$, the demand of $\mcone{r}$ is at most $(\floordjbyu{\dense} + res(j_d) +1)u  \leq (\floordjbyu{\dense} + 2)u$. The available opening is $\floordjbyu{\dense} + 1$. Thus, the capacity violation is at most $(\floordjbyu{\dense} + 2)u / (\floordjbyu{\dense} + 1)u <2$ as $ \floordjbyu{\dense} \geq 1$. Hence $\mcone{r}$ is self-sufficient.

		
		For $\sigma_r=0$, ($ii$) - ($vi$) hold vacuously. Thus, now onwards we assume that $\sigma_r \geq 1$ ($ii$) $LP_2$ opens $\max \{0, \sigma'_r - 1\} $ facilities in $\mctwo{r}$ where $\sigma'_r$ is the number of clusters in $\mctwo{r}$. Constraint (\ref{LP-meta-clusters_const1-sparse}) ensures that at most one facility is opened in each cluster. Thus, there is at most one cluster  in $\mctwo{r}$ with no facility opened in it and it is a sparse cluster. Next consider $\mcone{r}$ with a sparse cluster in it, i.e.,~$\mcone{r} = \{j_d, j_s\}$, 
		it is possible that all the $\gamma_r$ facilities are opened in $\T{j_d}$ and no facility is opened in $\T{j_s}$. Thus, there are at most two clusters with no facility opened in them and these clusters are sparse. 
		($iii$) \& ($iv$) Let $j' \in \mctwo{r}$ such that no facility is opened  in $\T{j'}$. If $\sigmaone(j') \in \mctwo{r}$, then $LP_2$ must have opened a facility in $\T{\sigmaone(j')}$. Demand of $j'$ is assigned to this facility at a loss of maximum $2$ factor in capacity. 
		If $\sigmaone(j') \notin \mctwo{r}$ then either $\mcone{r}$ is empty or $\sigmaone(j') \in \mcone{r}$. In the former case $j'$ must be the root of $\cen{r}$ and $\cen{r}$ cannot be the root MC. Clearly, at most $u$ units of demand of $\cen{r}$ remain unassigned within $\cen{r}$.
		In the latter case i.e.,~$\sigmaone(j')  \in \mcone{r}$, then $\sigmaone(j')$  is either $\dense$ or $\sparseone$. We will next show that demand of $j'$ will be absorbed in $\T{\dense} \cup \T{\sparseone}$ in the claimed bounds along with claims $(v)$ and $(vi)$ of the lemma.


		\begin{enumerate}
			\item $\resj{\dense} < \epsilon$, we have $ \mcone{r} = \{\dense\},~\gamma_r= \floordjbyu{\dense}$, $\mctwo{r} = \cen{r} \cap \csparse$, 
			$ \sigma'_r =  \sigma_r$, 
			and $\beta_r=\floordjbyu{\dense}+ \sigma_r -1$. 
			In this case, $j' = \sparseone$ and $\sigmaone(j') = \dense$. 
			$LP_2$ must have opened at least $\floordjbyu{\dense} \geq 1$ facilities in 
			$\T{\dense}$
			Total demand $(\floordjbyu{\dense} + res(j_d) + 1)) u$ of $j_d$ and $j'$ can be distributed to the  facilities  opened in $\T{j_d}$ ($\geq \floordjbyu{\dense}$) at a loss of factor $2 + \epsilon$ in capacity, as $res(j_d) < \epsilon$ and $1 \leq \floordjbyu{\dense}$.
			

			\item $\epsilon \leq \resj{\dense}  < 1$, we have $\mcone{r} = \{\dense,~\sparseone\}$,$~\gamma_r= \floordjbyu{\dense} + 1$,
			$\mctwo{r} = \cen{r} \cap \csparse \setminus \{\sparseone\} $, $ \sigma'_r =  \sigma_r - 1$ and 
			$\beta_r=\floordjbyu{\dense}+\sigma_r-1$ if $\sigma_r \geq 2$ and $=\floordjbyu{\dense}+1$ if $\sigma_r=1$.
			In this case, $\sigmaone(j') = \sparseone$.
			In the worst case, no facility is opened in $\T{\sparseone}$. 
			$LP_2$ must have opened at least $\floor{d_{\dense}/u} + 1 \geq 2$ facilities in $\T{\dense} \cup \T{\sparseone}$.
			Total demand $(\floordjbyu{\dense} + res(j_d) + 1 + 1)) u$ of $j_d, j_s$ and $j'$ can be distributed to the  facilities  opened in $\T{j_d} \cup \T{j_s}$ ($\geq \floordjbyu{\dense} + 1$) at a loss of factor $2$ in capacity, as  $ \floordjbyu{\dense} + 1 \geq 2$.
			
		\end{enumerate} 
		
		($vii$) Clearly, $\dist{j'}{\dense} \leq 2 \dist{j'}{\sigmaone(j')}$. ($2$) above also handles the case when no facility is opened in a sparse cluster in $\mcone{r}$.

\end{proof}
Lemma (\ref{io-in-grp-knm_new2}) deals with the remaining demand that we fail to assign  within the MC. Such demand is assigned in the parent MC. Lemma (\ref{io-in-grp-knm_new2}) discusses the cost bound for such assignments and the impact of the demand coming onto $\cen{r}$ from the children MCs along with the demand within $\cen{r}$ on capacity.

\begin{lemma}
	\label{io-in-grp-knm_new2}
	Consider a meta-cluster $\cen{r}$. The demand of $\cen{r}$ and the demand coming onto $\cen{r}$ from the children meta-clusters can be assigned to the facilities opened in $\cen{r}$ such that: (i) capacities are violated at most by a factor of ($2+\frac{4}{\ell-1}$) for $\ell \geq 1/\epsilon$ and, (ii) Total distance traveled by demand $\demandofj{j'}$ of $j' \in \cliset'$ to reach the centers of the clusters in which they are served is bounded by $\ell d_{j'}\dist{j'}{\sigmaone(j')}$.
\end{lemma}
\begin{proof}
	Proof is similar to the proof of Lemma (\ref{io-in-grp-knm_new2.1}).
\end{proof}

\begin{lemma}
	\label{bd-cost-to-facility-1}
	The cost of assigning the demands collected at the centers to the facilities opened in their respective clusters is bounded by $(2+\epsilon)(2\ell+1)LP_{opt}$.
\end{lemma}
\begin{proof}
	Proof is similar to the proof of Lemma (\ref{bd-cost-to-facility}).
\end{proof}

	\section{Capacitated $k$ Facility Location Problem} \label{section_cap$k$FLP}
	Standard LP-Relaxation of the C$k$FLP can be found in Aardal \etal~\cite{capkmGijswijtL2013}. 
	When $\facilitycost = 0$, the problem reduces to the $k$-median problem and when $k = \left| \facilityset \right| $ it reduces to the facility location problem. Our techniques for CKnM provide similar results for C$k$FLP in a straight forward manner i.e.,~O($1/\epsilon^2$) factor approximation, violating the capacities by a factor of $(2 + \epsilon)$ and cardinality by plus $1$. 
	The violation of cardinality can be avoided  by opening the facility with larger opening integrally while converting a pseudo integral solution into an integrally open solution.  
	Thus, we obtain Theorem~\ref{theorem2}.
	
	\textbf{Proof of Theorem 3:} Let $\rho^{*} = <x^*, y^*>$ denote the optimal $LP$ solution. For sparse clusters, we open the cheapest facility $i^*$ in $\ballofj{j}$, close all facilities in the cluster and shift their demands to $i^*$.
	Let  $\hat{\rho} = <\hat{x}, \hat{y}>$ be the solution so obtained. It is easy to see that we loose at most a factor of $2$ in cardinality, and $\mathcal{C}ost$k$FLP(\hat{x}, \hat{y})$ is within $O(1) LP_{opt}$. 
	
	To handle dense clusters, we introduce the notion of cluster instances. For each cluster center $j' \in \crich$, let  ${\pricef{j'}} = \sum_{i \in \bundle{j'}} \facilitycost y^*_i$ and $ \price{j'} = \sum_{i \in \bundle{j'}} \sum_{ j \in \clientset} x^*_{i j} \lbrack \dist{i}{ j} +4 \C{ j} \rbrack $.
	We define a cluster instance $\starinst{S}{j'}{\bundle{j'}}{\demandofj{j'}}{\price{j'},{\pricef{j'}}}$ as follows: Minimize $Cost_{CI}(z) = \sum_{\singlefacility \in \bundle{j'}} ( f_i +\capacity \dist{i}{j'} ) \zofi{i}$ s.t. $\capacity \sum_{\singlefacility \in \bundle{j'}} \zofi{i} \geq \demandofj{j'}$ and $z_i \in [0,1]$. 
	It can be shown that $\zofi{i} = \sum_{ j \in \clientset}x_{i j}^*/\capacity = l_i/\capacity \le y^*_i\ \forall i \in \bundle{j'}$ is a feasible solution with cost at most $\pricef{j'} + \price{j'}$. An almost integral solution $z'$ is obtained by arranging the fractionally opened facilities in $z$ in non-decreasing order of $\facilitycost  + \dist{i}{j'} \capacity $ and greedily transferring the total opening $\cardTwo{z}{\bundle{j'}}$ to them. Let $ l_i' = \primezofi{i}  \capacity $. For a fixed $\epsilon > 0$, an integrally open solution $\hat{z}$ and assignment $\hat{l}$ (possibly fractional) is obtained as follows: let $i_1$ be the fractionally opened facility, if any. If  $\primezofi{i_1} < \epsilon$, close $i_1$ and shift its demand to another integrally opened facility at a loss of factor $(1+\epsilon)$ in its capacity. Else ($\primezofi{i_1} \geq \epsilon$), open $i_1$, at a loss of factor $2$ in cardinality and $1/\epsilon$ in facility cost. The solution $\hat{z}$ satisfies the following: $\hat{l}_i \leq ( 1 + \epsilon) \hat{z}_i u \ \forall i \in \bundle{j'}$, $\sum_{i \in \bundle{j'}} \hat{z}_i \le 2 \sum_{i \in \bundle{j'}} z'_i \ \forall  j' \in\ \crich$ and $Cost_{CI} (\hat{z}) \leq \max\{1/\epsilon, 1 + \epsilon\} Cost_{CI} (\hat{z})$.
	

\section{Conclusion}
\label{conclusion}
In this work, we presented the first constant factor approximation algorithm for uniform hard capacitated knapsack median problem violating the budget by a factor of ($1+\epsilon$) and capacity by ($2+\epsilon$). Two variety of results were presented for capacitated $k$-facility location problem with a trade-off between capacity and cardinality violation: an $O(1/\epsilon^2)$ factor approximation violating capacities by ($2+\epsilon$) and a $O(1/\epsilon)$ factor approximation, violating the capacity by a factor of at most $(1 + \epsilon)$ using at most $2k$ facilities.
As a by-product, we also gave a constant factor 
approximation for uniform capacitated facility location at a loss of $(1+\epsilon)$ in capacity from the natural LP. The result shows that the natural LP is not too bad.


It would be interesting to see if the capacity violation can be reduced to ($1 + \epsilon$) using the techniques of Byrka \etal~\cite{ByrkaRybicki2015}. Avoiding violation of budget will require strengthening the LP in a non-trivial way.
Another direction for future work would be to extend our results to non-uniform capacities.  Conflicting requirement of facility costs and capacities makes the problem challenging. 
	


\bibliography{ref_master}

\begin{thebibliography}{10}

\bibitem{capkmGijswijtL2013}
Karen Aardal, Pieter~L. van~den Berg, Dion Gijswijt, and Shanfei Li.
\newblock Approximation algorithms for hard capacitated k-facility location
  problems.
\newblock {\em EJOR}, 242(2):358--368, 2015.
\newblock \href {http://dx.doi.org/10.1016/j.ejor.2014.10.011}
  {\path{doi:10.1016/j.ejor.2014.10.011}}.

\bibitem{mathp}
Ankit Aggarwal, Anand Louis, Manisha Bansal, Naveen Garg, Neelima Gupta,
  Shubham Gupta, and Surabhi Jain.
\newblock A 3-approximation algorithm for the facility location problem with
  uniform capacities.
\newblock {\em Journal of Mathematical Programming}, 141(1-2):527--547, 2013.

\bibitem{AnMP2015}
Hyung{-}Chan An, Aditya Bhaskara, Chandra Chekuri, Shalmoli Gupta, Vivek Madan,
  and Ola Svensson.
\newblock Centrality of trees for capacitated k-center.
\newblock {\em Math. Program.}, 154(1-2):29--53, 2015.
\newblock URL: \url{https://doi.org/10.1007/s10107-014-0857-y}, \href
  {http://dx.doi.org/10.1007/s10107-014-0857-y}
  {\path{doi:10.1007/s10107-014-0857-y}}.

\bibitem{Anfocs2014}
Hyung-Chan An, Mohit Singh, and Ola Svensson.
\newblock Lp-based algorithms for capacitated facility location.
\newblock In {\em FOCS, 2014}, pages 256--265.

\bibitem{Bansal}
Manisha Bansal, Naveen Garg, and Neelima Gupta.
\newblock A 5-approximation for capacitated facility location.
\newblock In {\em Algorithms - {ESA} 2012 - 20th Annual European Symposium,
  Ljubljana, Slovenia, September 10-12, 2012. Proceedings}, pages 133--144,
  2012.
\newblock \href {http://dx.doi.org/10.1007/978-3-642-33090-2_13}
  {\path{doi:10.1007/978-3-642-33090-2_13}}.

\bibitem{capkmByrkaFRS2013}
Jaroslaw Byrka, Krzysztof Fleszar, Bartosz Rybicki, and Joachim Spoerhase.
\newblock Bi-factor approximation algorithms for hard capacitated
  \emph{k}-median problems.
\newblock In {\em {SODA} 2015}, pages 722--736.

\bibitem{Byrkaesa2015}
Jaroslaw Byrka, Thomas Pensyl, Bartosz Rybicki, Joachim Spoerhase, Aravind
  Srinivasan, and Khoa Trinh.
\newblock An improved approximation algorithm for knapsack median using
  sparsification.
\newblock In {\em {ESA} 2015}, pages 275--287.

\bibitem{ByrkaRybicki2015}
Jaroslaw Byrka, Bartosz Rybicki, and Sumedha Uniyal.
\newblock An approximation algorithm for uniform capacitated k-median problem
  with {(1} + {\(\epsilon\)}) capacity violation.
\newblock In {\em Integer Programming and Combinatorial Optimization - 18th
  International Conference, {IPCO} 2016, Li{\`{e}}ge, Belgium, June 1-3, 2016,
  Proceedings}, pages 262--274, 2016.
\newblock \href {http://dx.doi.org/10.1007/978-3-319-33461-5_22}
  {\path{doi:10.1007/978-3-319-33461-5_22}}.

\bibitem{charikar}
Moses Charikar and Sudipto Guha.
\newblock Improved combinatorial algorithms for the facility location and
  k-median problems.
\newblock In {\em Proceedings of the 40th Annual IEEE Symposium on Foundations
  of Computer Science (FOCS), New York, NY, USA}, pages 378--388, 1999.

\bibitem{charikar2005improved}
Moses Charikar and Sudipto Guha.
\newblock Improved combinatorial algorithms for facility location problems.
\newblock {\em SIAM Journal on Computing}, 34(4):803--824, 2005.

\bibitem{Charikar:1999}
Moses Charikar, Sudipto Guha, {\'{E}}va Tardos, and David~B. Shmoys.
\newblock A constant-factor approximation algorithm for the \emph{k}-median
  problem (extended abstract).
\newblock In {\em {STOC}, 1999}, pages 1--10.
\newblock \href {http://dx.doi.org/10.1145/301250.301257}
  {\path{doi:10.1145/301250.301257}}.

\bibitem{Charikaricalp2012}
Moses Charikar and Shi Li.
\newblock A dependent lp-rounding approach for the k-median problem.
\newblock In {\em {ICALP} 2012}, pages 194--205.

\bibitem{CyganFOCS2012}
Marek Cygan, MohammadTaghi Hajiaghayi, and Samir Khuller.
\newblock {LP} rounding for k-centers with non-uniform hard capacities.
\newblock pages 273--282, 2012.
\newblock URL: \url{https://doi.org/10.1109/FOCS.2012.63}, \href
  {http://dx.doi.org/10.1109/FOCS.2012.63} {\path{doi:10.1109/FOCS.2012.63}}.

\bibitem{Demirci2016}
H.~G{\"{o}}kalp Demirci and Shi Li.
\newblock Constant approximation for capacitated k-median with {(1} +
  {\(\epsilon\)})\emph{} capacity violation.
\newblock In {\em 43rd {ICALP} 2016}, pages 73:1--73:14.

\bibitem{k-centerKhuller2000}
Samir Khuller and Yoram~J. Sussmann.
\newblock The capacitated \emph{K}-center problem.
\newblock {\em {SIAM} J. Discrete Math.}, 13(3):403--418, 2000.
\newblock URL: \url{https://doi.org/10.1137/S0895480197329776}, \href
  {http://dx.doi.org/10.1137/S0895480197329776}
  {\path{doi:10.1137/S0895480197329776}}.

\bibitem{KPR}
Madhukar~R. Korupolu, C.~Greg Plaxton, and Rajmohan Rajaraman.
\newblock Analysis of a local search heuristic for facility location problems.
\newblock {\em Journal of Algorithms}, 37(1):146--188, 2000.

\bibitem{Krishnaswamysoda2011}
Ravishankar Krishnaswamy, Amit Kumar, Viswanath Nagarajan, Yogish Sabharwal,
  and Barna Saha.
\newblock The matroid median problem.
\newblock In {\em {SODA}, 2011}, pages 1117--1130.

\bibitem{krishnaswamySTOC18}
Ravishankar Krishnaswamy, Shi Li, and Sai Sandeep.
\newblock Constant approximation for k-median and k-means with outliers via
  iterative rounding.
\newblock In {\em Proceedings of the 50th Annual {ACM} {SIGACT} Symposium on
  Theory of Computing, {STOC} 2018, Los Angeles, CA, USA, June 25-29, 2018},
  pages 646--659, 2018.
\newblock URL: \url{http://doi.acm.org/10.1145/3188745.3188882}, \href
  {http://dx.doi.org/10.1145/3188745.3188882}
  {\path{doi:10.1145/3188745.3188882}}.

\bibitem{Amitsoda2012}
Amit Kumar.
\newblock Constant factor approximation algorithm for the knapsack median
  problem.
\newblock In {\em {SODA}, 2012}, pages 824--832.

\bibitem{LeviSS12}
Retsef Levi, David~B. Shmoys, and Chaitanya Swamy.
\newblock Lp-based approximation algorithms for capacitated facility location.
\newblock {\em Journal of Mathematical Programming}, 131(1-2):365--379, 2012.

\bibitem{capkmshanfeili2014}
Shanfei Li.
\newblock {An Improved Approximation Algorithm for the Hard Uniform Capacitated
  k-median Problem}.
\newblock In {\em Approximation, Randomization, and Combinatorial Optimization.
  Algorithms and Techniques (APPROX/RANDOM 2014), Germany}, pages 325--338.

\bibitem{capkmshili2014}
Shi Li.
\newblock On uniform capacitated \emph{k}-median beyond the natural {LP}
  relaxation.
\newblock In {\em Proceedings of the Twenty-Sixth Annual {ACM-SIAM} Symposium
  on Discrete Algorithms, {SODA} 2015, San Diego, CA, USA, January 4-6, 2015},
  pages 696--707, 2015.

\bibitem{Lisoda2016}
Shi Li.
\newblock Approximating capacitated \emph{k}-median with {(1} +
  {\(\epsilon\)})\emph{k} open facilities.
\newblock In {\em Proceedings of the Twenty-Seventh Annual {ACM-SIAM} Symposium
  on Discrete Algorithms, {SODA} 2016, Arlington, VA, USA, January 10-12,
  2016}, pages 786--796, 2016.

\bibitem{Lin92}
Jyh{-}Han Lin and Jeffrey~Scott Vitter.
\newblock epsilon-approximations with minimum packing constraint violation
  (extended abstract).
\newblock In {\em Proceedings of the 24th Annual {ACM} Symposium on Theory of
  Computing, May 4-6, 1992, Victoria, British Columbia, Canada}, pages
  771--782, 1992.
\newblock \href {http://dx.doi.org/10.1145/129712.129787}
  {\path{doi:10.1145/129712.129787}}.

\bibitem{Shmoys}
David~B. Shmoys, {\'E}va Tardos, and Karen Aardal.
\newblock Approximation algorithms for facility location problems (extended
  abstract).
\newblock In {\em Proceedings of the Twenty-Ninth Annual ACM Symposium on the
  Theory of Computing, El Paso, Texas, USA,}.

\bibitem{Swamyapprox2014}
Chaitanya Swamy.
\newblock Improved approximation algorithms for matroid and knapsack median
  problems and applications.
\newblock In {\em {APPROX/RANDOM}, 2014}, pages 403--418, 2014.

\end{thebibliography}
\end{document}